\documentclass[11pt,reqno]{amsart}

\usepackage{amsmath, amsthm, amscd, amsfonts, amssymb, graphicx, color, mathrsfs, extarrows}
\usepackage{mathscinet}
\usepackage[a4paper,left=3cm,right=3cm,top=2cm,bottom=4cm,bindingoffset=5mm]{geometry}
\usepackage{fancyhdr}
\usepackage{comment}
\usepackage{float}
\usepackage{diagbox}
\usepackage{eurosym}
\usepackage[utf8]{inputenc}
\usepackage{amsmath}
\usepackage{amsfonts}
\usepackage{version}
\usepackage{mathtools}
\usepackage{amscd}
\usepackage{amssymb}
\usepackage{dsfont}
\usepackage{amsthm}
\usepackage{thmtools}
\usepackage{graphicx}
\usepackage{mathrsfs}
\usepackage{todonotes}
\usepackage{xcolor}
\usepackage{todonotes}
\usepackage{setspace}
\usepackage{enumerate}
\usepackage[english]{babel}
\usepackage{xcolor}
\usepackage[colorlinks=true,linkcolor=blue]{hyperref}
\usepackage[all]{hypcap}

\newtheorem{theorem}{Theorem}

\theoremstyle{definition}

\newcommand{\R}{\mathbb{R}}
\newcommand{\N}{\mathbb{N}}
\renewcommand{\P}{{\mathbb P}}

\newcommand{\E}{\mathbb{E}}

\newcommand{\PD}{{\rm PD}}

\newcommand{\eins}{\mathds{1}}

\newcommand{\old}{{\rm old}}
\newcommand{\class}{{\rm grade}}
\newcommand{\Cov}{{\rm Cov}}

\newcommand{\DR}{{\rm DR}}
\newcommand{\LTDR}{{\rm LRDR}}
\newcommand{\LTCT}{{\rm LRCT}}

\newcommand{\RD}{{\rm RD}}

\title[Testing for long-term calibration in rating systems]{A hypothesis test for the long-term calibration in rating systems with overlapping time windows}

\author{Patrick Kurth}
\address{Landesbank Baden-W\"urttemberg, Stuttgart, Germany}
\email{patrick.kurth@lbbw.de}

\author{Max Nendel}
\address{Center for Mathematical Economics, Bielefeld University, Germany}
\email{max.nendel@uni-bielefeld.de}

\author{Jan Streicher}
\address{Center for Mathematical Economics, Bielefeld University, Germany and Landesbank Baden-W\"urttemberg, Stuttgart, Germany}
\email{jan.streicher@uni-bielefeld.de}

\thanks{The authors thank Markus Klein for helpful discussions related to this work.\ This work was funded by the Deutsche Forschungsgemeinschaft (DFG, German Research Foundation) -- SFB 1283/2 2021 -- 317210226.\ The first and the third author are grateful for the support of the Landesbank Baden-W\"urttemberg related to this work.}

\date{\today}

\begin{document}
\maketitle

\begin{abstract}
We present a statistical test that can be used to verify supervisory requirements concerning overlapping time windows for the long-term calibration in rating systems.\ In a first step, we show that the long-run default rate is approximately normally distributed with respect to random effects in default realization. We then perform a detailed analysis of the correlation effects caused by the overlapping time windows and solve the problem of an unknown distribution of default probabilities for the long-run default rate. In this context, we present several methods for a conservative calibration test that can deal with the unknown variance in the test statistic. We present a test for individual rating grades, and then pass to the portfolio level by suitably adapting the test statistic.\ We conclude with comparative statics analysing the effect of persisting customers and the number of customers per reference date.
\end{abstract}


\section{Introduction}

Financial institutions use statistical models to estimate the default risk of obligors in order to manage credit-risks. According to Basel II, banks are allowed to estimate risk parameters that are used to calculate regulatory capital with their own models. The legal framework for the use of such models in the internal ratings-based (IRB) approach is regulated in the Capital Requirements Regulation (CRR), see \cite{CRR}.\ The CRR imposes specific requirements for the models, e.g., that \qq{institutions shall estimate PDs by obligor grade from long run averages of one-year default rates}, see Article 180.\ In 2017, the European banking authority published the \qq{guidelines on PD estimation, LGD estimation and the treatment of defaulted exposures} (EBA-GL) specifying the CRR requirement, see \cite{EBA_GL}.\ Paragraph 81 EBA-GL states that \qq{institutions should calculate the observed average default rates as the arithmetic average of all one year default rates}. This requirement was additionally specified in \cite{ECB Guide} and in \cite{RTS (EBA)}. Here, the observed one-year default rate at a given reference date is defined as the percentage of defaulters in the following year, so that the observed long-run average default rate depends on the choice of the reference dates.\ Paragraph 80 EBA-GL allows institutions to choose \qq{between an approach based on overlapping and an approach based on non-overlapping one-year time windows}.\ Overlapping one-year time windows occur when the time interval between two reference dates is less than one year.\ Due to computational simplicity, it is of course convenient to continue working with non-overlapping time windows and appropriately adjust the long-term default rate.\ In many cases, these type of adjustments are rather on the conservative side, see, e.g., \cite{jing2008asset} for an empirical study and \cite{li2016probability,zhou2001analysis} for theoretical analyses in the context of asset correlation.\

On the other hand, the approach, using overlapping time windows, provides more information on defaults due to potential short-term contracts, which cannot be observed during one-year periods.\ It is therefore favored by most financial institutions that handle portfolios with only few observed defaults.\  
Another advantage of this approach lies in the fact that the bias caused by a specific choice of reference dates can be reduced, e.g., when calculating the long-run average default rate as the arithmetic mean of the one-year default rates on a quarterly basis.

Paragraph 87 EBA-GL requires institutions to use a statistical test of calibration at the level of rating grades and for certain portfolios.\ Classically, the literature dealing with calibration tests assumes a binomial-distributed default rate, see, e.g., \cite{bundesbank2003approaches} and \cite{tasche2008validation} for a discussion of different hypothesis tests in this context. We also refer to \cite{coppens2016advances} for the consideration of different PDs within the same rating grade and \cite{blochwitz2004validating,blochwitz2006statistical} for a modified binomial test  accounting correlated defaults. However, when considering overlapping time windows, the assumption of a binomial distribution can no longer be maintained when considering overlapping time windows.

More generally, Monte Carlo methods can be used to construct tests for distributions that cannot be determined analytically. For the use of Monte Carlo methods, precise knowledge of the distribution of the probabilities of default within the portfolio is essential.\ However, in our case, these probabilities are unknown and estimated by the underlying model, so that they cannot not be used to determine the test statistic. On the other hand, analytical tests are desirable since they satisfy requirements such as replicability and reproducibility. We refer to \cite{blochlinger2012validation,blochlinger2016probabilities} for hypothesis tests in analytic form that take into account correlation effects, replacing the i.i.d.\ assumption by a conditional i.i.d.\ assumption. In this context, we also refer to \cite{tasche2003traffic} for a one-observation-based inference on the adequacy of probability of default forecasts with dependent default events.

In this paper, we therefore present a statistical test that can be used to verify supervisory calibration requirements in the case of overlapping time windows. A major challenge lies in the analysis of the correlation effects that are caused by the overlapping time windows. On the level of individual ratings grades, the variance of the test statistic is already determined by the null hypothesis. This, however, is not the case on a portfolio level, which is why we focus in great detail on a conservative estimate for the variance by solving a related minimization problem. We then present a conservative calibration test that can deal with the unknown variance in the test statistic.

The rest of the paper is structured as follows.\ In Section \ref{sc: notations}, we introduce the terminology and notation to for the formulation of the hypothesis test in Section \ref{Hypothesis Test}. In Section \ref{Normal Distribution of long run default rate}, we show that the long-run default rate is approximately normally distributed with respect to random effects in default realization. Thereafter, we focus on the analysis of the correlation effects that arise due to the overlapping time periods in Section \ref{Covariance between Default States} and Section \ref{Covariance between Default Rates}. Taking these correlation effects into account is essential for determining the variance of the long-run default rate, see Section \ref{Distribution of the long-run default rate}. \ In Section \ref{hypothesis test for long-term calibration} we formulate the hypothesis test in detail, first at the level of individual rating grades, cf.\ Section \ref{hypothesis test rating grade} and then at portfolio level, cf.\ Section \ref{Statistical Test on Portfolio Level}. We conclude with a discussion on the parameters of the test and further considerations in Section \ref{sec:4}.\ A closed form solution to the minimization problem related to the estimation of the variance is derived in the Appendix \ref{ Minimization problem}. In the Appendix \ref{Appendix E}, we propose an alternative way to estimate the variance without solving and optimization problem.

\section{Setup and preliminaries} \label{sc: setup}

\subsection{Setup and notation} \label{sc: notations}
In this section, we introduce the terminology and notation used in this paper and give a formal description of the statistical test that is studied in this work.\
We begin with the general setup by considering the default state of an individual obligor. Throughout, a default state over a one-year time horizon is described by a Bernoulli-distributed random variable $x\sim B\left(1,p\right)$, where
$$p\in\big\{\PD_1,\ldots,\PD_m \big\},$$ 
$m\in \N$, and $0<\PD_1<\ldots<\PD_m<1$ are the default probabilities of the rating grades in the underlying master scale. In the following, we also use the notation $\PD_{\min}$ and $\PD_{\max}$ for the default probabilities $\PD_1$ and $\PD_m$, respectively.

Next, we introduce the long-run default rate for a given history of reference dates $\RD_1<\ldots<\RD_N$ with $N\in \N$. For all $t=1,\ldots,N$, we are given a number $n_t\in \N_0$ of customers within the portfolio at the reference date $\RD_t$, and write 
\[
n_{\min}:=\min \big(\{n_t\, |\, t=1,\ldots,N \}\setminus\{0\}\big)\quad\text{and}\quad n_{\max}:=\max \{n_t\, |\, t=1,\ldots,N \}
\]
for the minimal and maximal number of existing customers at any reference date throughout the history of the portfolio, respectively.\ Throughout, we assume that $n_{\max}\neq 0$ or, equivalently, that $n_t\neq 0$ for some $t=1,\ldots, N$.

Let $q\in \N$ be the number of reference dates within a one-year time-horizon starting from an arbitrary reference date.\ To be more precise, we assume that
$$
\RD_{t+1}=\RD_t+\frac{1}q\quad\text{for all }t= 1,\ldots, N-1.
$$
According to the EBA guidelines \cite{EBA_GL}, the long-run default rate should be computed at least on a quarterly basis, i.e., $q=4$, or on an annual basis, i.e., $q=1$, performing however an analysis of the possible bias that occurs due to the negligence of quarterly data.\ The main interest of our analysis lies in the case $q>1$ leading to correlation effects caused by overlapping time windows. For $s,t=1,\ldots,N$ with $s>t$, we define $w_{t,s}:=\max\{0,\RD_{t}+1-\RD_{s}\}$ for the size of the overlap of the observation periods.

Since new customers are usually added over time and business relationship ends for others, assign to each customer during the history of reference dates a number from $1$ to $M\in \N$, and we define $\Lambda_t\subseteq\{1,..,M\}$ as the set of obligors at reference date $\RD_t$ for all $t=1,\ldots, N$. In particular, $\bigcup_{t=1}^N \Lambda_t=\{1,\ldots, M\}$, $\left|\Lambda_t\right|=n_t$, and $n_{\max}\leq M$.\ For $t=1,\ldots, N$, we then define the one-year default rate among the considered customers at reference date $\RD_t$ by
\begin{align}\label{defaulte rate RD}
 X_t=\frac{1}{n_t}\sum\limits_{j\in\Lambda_t}x_{j,t},   
\end{align}
where $x_{t,j}\sim B(1,p_{t,j})$ is the one-year default state and $p_{t,j}\in \left\{\PD_1,\ldots,\PD_m \right\} $ the probability of default over a one year time horizon of customer $j\in \Lambda_t$ at the reference date $\RD_t$. Since $n_t=0$ cannot be excluded, we use the convention $\frac{0}{0}:=0$.

The realized default rate on a reference date $\RD_t$ is denoted by $\DR_t$, i.e., $\DR_t$ is the realization of the random variable $X_t$ for $t=1,\ldots, N$. Moreover, we define 
$$R_N:=\big\{t=1,\ldots, N\,|\, n_t> 0\big\}$$ 
as the set of indices for reference dates, where the portfolio contains at least one customer, and denote its cardinality by $R(N)$.\ Since we assume that $n_{\max}\neq 0$, it follows that $R(N)>0$, and we define the realized long-run default rate as
$$\LTDR:=\frac{1}{R(N)} \sum_{t=1}^{N}\DR_t,$$
which is a realization of the the long-run default rate
\begin{align} \label{long-run DR}
  Z:=\frac{1}{R(N)}\sum_{t=1}^{N}X_t.
\end{align}
  We emphasize that the long-run default rate is the arithmetic mean of the one-year default rates and \textit{not} the arithmetic mean of the individual default states of all customers for all reference dates, which intuition might suggest. Due to the chosen convention $\frac{0}{0}=0$, it follows that $X_t=0$ if $n_t=0$, i.e., reference dates, where the portfolio contains no customers are completely disregarded in the computation of the long-run default rate. Nevertheless, these reference dates have to be included into the timeline since they describe an overlapping period. 

  \subsection{Formal description of the test}\label{Hypothesis Test}
We now give an overview over the hypothesis test presented in this paper.
The aim is to formulate a statistical test for comparing the realized long-run default rate $\LTDR$ with the estimated long-run default rate $\LTCT$, which is called the long-run central tendency, both on the level of individual rating grades and on portfolio level.\ We use the random variable $Z$ with expected value $\mu$ as the test statistic. We formulate the following hypothesis test with
\begin{align*}
 & \text{null hypothesis }H_0:\mu=\LTCT \quad \text{and}\\
 & \text{alternative hypothesis }H_1:\mu\neq \LTCT.
\end{align*}
Based on this, we determine values $k,K\in [0,1]$ and consider the test function
\begin{equation}\label{eq:phi}
\varphi(z):=\begin{cases}
 0,& \text{if}\ z\in [k,K],\\
 1,& \text{otherwise}.
\end{cases}
\end{equation}
 If $\varphi \left( \LTDR \right)=0$, the null hypothesis based on the available data is retained, whereas it is rejected in favor of the alternative hypothesis if $\varphi\left(\LTDR\right)=1$.\ We point out that the distribution of the test statistic $Z$ does not only depend on $\mu$ but also on other parameters that are not determined by the null hypothesis. The main focus of the paper is to postulate distributional assumption for $Z$ such that its distribution is determined by the null hypothesis. In this context, we first show that $Z$ is approximately normally distributed, i.e., $Z\sim \mathcal{N} \left(\mu,\sigma^2\right)$.\ In a second step, the focus lies on a conservative estimate for the variance $\sigma^2$ based on the information given by the null hypothesis. A major challenge lies in the fact that the consideration of overlapping time windows leads to unavoidable correlation effects that rule out independence assumptions on the random variables $X_1,\ldots, X_N$.

\subsection{Distribution of the long-run default rate} \label{Normal Distribution of long run default rate}

As we have seen in the previous section, the long-run calibration test requires the formulation of a distribution assumption for the long-run default rate. However, the definition of the long-run default rate as the arithmetic mean of the one-year default rates leads to difficulties in the derivation of an analytical description of the distribution based on the Bernoulli-distributed default states of the individual obligors.\ The aim of this section is to show that, despite the correlation effects and the possibly varying number of obligors with different PDs at each reference date, the long-run default rate is still approximately normally distributed.

In order to simplify notation, we define $y_{t,j}:=x_{t,j}$ if $j\in\Lambda_t$ and $y_{t,j}:=0$ otherwise, for all $t=1,\ldots, N$. Then,
$$
Z=\frac{1}{R(N)}\sum\limits_{t=1}^{N}\frac{1}{n_t}\sum\limits_{j=1}^My_{t,j}=\sum\limits_{j=1}^M Y_j
$$
with $Y_j=\sum\limits_{t=1}^{N}\frac{1}{R(N)}\frac{1}{n_t} y_{t,j}.$

In the sequel, we will show that $Z$ is approximately normally distributed. For this, we assume that default states of different obligors are independent of each other, from which the independence of the family $(Y_j)_{j=1,\dots,M}$ follows. This assumption is standard and sufficiently conservative for the calibration of rating models.

In order to apply the Lindeberg-Feller central limit theorem (CLT), which is a generalization of the classical CLT, we first show that the variances of the random variables $(Y_j)_{j=1,\ldots, M}$ cannot become arbitrarily small. In fact, by neglecting the covariance between the default states of an obligor at different reference dates,
\begin{align*}
\sigma^2\left(Y_j\right)&=\frac{1}{R(N)^2}\sum\limits_{t=1}^{N}\bigg(\frac{1}{n_t^2}\sigma^2(y_{t,j})+2\frac{1}{R(N)^2}\sum\limits_{s=t+1}^N\frac{1}{n_{t}n_{s}}\Cov(y_{t,j},y_{s,j})\bigg)\\
&\geq\frac{1}{R(N)^2}\sum\limits_{t=1}^{N}\frac{1}{n_t^2}\sigma^2(y_{t,j})\quad\text{for all } j=1,\ldots,M.
\end{align*}
Moreover, for each obligor $j=1,\ldots,M$, there exists at least one index $t=1,\ldots,N$ with $\sigma^2(y_{t,j})>0$ since  $\bigcup_{t=1}^N \Lambda_t=\{1,\ldots, M\}$. Hence, for each obligor $j=1,\ldots, M$,
$$
\sigma^2\left(Y_j\right)\geq\frac{1}{R(N)}\frac{1}{n_{\max}^2}\min\limits_{k=1,\ldots,m}\left\{(\PD_k(1-\PD_k))\right\}=:v_{\min}.
$$
 Note that $v_{\min}$ is independent of the customer $j=1,\ldots, M$ and $v_{\min}>0$.\ Since, up to now, we only consider finitely many customers, we extend the family $(Y_j)_{j=1,\dots,M}$ by a sequence of independent random variables $Y_{M+1},Y_{M+2},\dots$ with $\sigma^2\left(Y_j\right)\geq v_{\min}$, e.g. $Y_j\sim \mathcal{N}\left(\mu_j,v_{\min}\right)$ with arbitrary drift $\mu_j\in(0,1)$ for $j\geq M+1$.
 
 Now, the central limit theorem shall be applied to the family $(Y_j)_{j\in\mathbb{N}}$.\ In the classical version by Lindeberg \& Lévy, the CLT states that a proper renormalization of the arithmetic mean is approximately normally distributed if the underlying sequence of random variables is independent and identically distributed. Its generalized version by Lindeberg \& Feller also applies to random variables that are not identically distributed as in our case.\ It states that a proper renormalization of the arithmetic mean of a family of independent random variables converges in probability against a normally distributed random variable if this family satisfies the so-called Lindeberg condition, see \cite[Definition 15.41]{klenke2020}.\ In our case, the Lindeberg condition applies if, for all $\varepsilon>0$,
$$
\lim_{k\rightarrow\infty}{\frac{1}{s_k^2}\sum_{j=1}^{k}\mathbb{E}\Big(\big(Y_j-\mathbb{E}(Y_j)\big)^2 \cdot\eins_{\big\{\left|Y_j-\mathbb{E}(Y_j)\right|>\varepsilon s_k\big\}}\Big)=0},
$$
where
$$
s_k=\sqrt{\sum_{j=1}^{k}{\sigma^2\left(Y_j\right)}}.
$$
That is, the Lindeberg condition requires that the underlying sequence of random variables does not exhibit arbitrarily large deviations from the expected value.

In the following, we verify the Lindeberg condition in the our setup.\ To that end, let $\varepsilon>0$. Since, for all $j=1,\ldots, M$,  the random variable $Y_j$ only takes values in $(0,1)$,
$$
\sigma^2\left(Y_j\right)\geq v_{\min}>0, \quad\text{and}\quad s_k\rightarrow\infty\quad\text{as }k\to \infty,
$$ 
 there exists some $k_0\in \mathbb{N}$ such that, for all $k\in \N$ with $k\geq k_0$, $j\in \mathbb N$, and $\omega\in \Omega$,
$$
|Y_j(\omega)-\mathbb{E}(Y_j)|\leq \varepsilon\cdot s_k.
$$
Hence, for all $k\in \N$ with $k\geq k_0$,
    \begin{align*}
   \frac{1}{s_k^2}\sum _{j=1}^{k}\mathbb{E}\Big(&\big(Y_j-\mathbb{E}(Y_j)\big)^2\cdot\eins_{\left\{ \left| Y_j-\mathbb{E}(Y_j) \right|>\varepsilon s_k\right\}}\Big)\\
   &=\frac{1}{s_k^2}\sum _{j=1}^{k_0}\mathbb{E}\left( \big( Y_j-\mathbb{E}(Y_j) \big)^2\cdot\eins_{\left\{ \left|Y_j-\mathbb{E}(Y_j) \right|>\varepsilon s_k\right\}}\right) \\ &\leq \frac{1}{s_k ^2}\sum_{j=1}^{k_0}\mathbb{E}\left( \big(Y_j-\mathbb{E}(Y_j)\big)^2\right) \rightarrow0\quad\text{as }k\to \infty,
    \end{align*} 
 i.e., the Lindeberg condition is satisfied. By the central limit theorem, cf.\ \cite[Theorem 15.44]{klenke2020}, we obtain that
$$ \lim_{k\rightarrow\infty}{\P\Bigg( \frac{1}{s_k}\sum _{j=1}^{k}\big( Y_j-\mathbb{E}(Y_j)\big) \leq z\Bigg)}=\Phi\left( z\right)\quad\text{for all }z\in \R,$$
where $\Phi$ is the cumulative distribution function for the standard normal distribution.\ Since the number of customers $M$ is large but fixed throughout our analysis, we may therefore assume that
\begin{equation} \label{eq:sim}
Z\sim \mathcal{N} \left(\mu,\sigma^2\right)
\end{equation}
with $\mu=\sum_{j=1}^M \E(Y_j)$ and $\sigma=s_M$, suppressing the dependence of $\mu$ and $\sigma$ on the number of customers $M$.

The literature suggests that, in standard cases, a sample size of at least $30$ is required in order to achieve useful results with the approximation of a normal distribution, see, for example, \cite{hoggtanis}.\ In our setting, however, the situation is more complex due to the specific calculation method of the long-run default rate.\ As we will see in Section \ref{sc: rate of convergence}, the convergence depends both on the number of reference dates $N$ and the number of customers $n_t$ per reference date $\RD_t$, hence indirectly on $M$, bot not solely on $M$. As a consequence, a large number of customers does not necessarily imply a good approximation, e.g., if we have $N=2$ reference dates and $M=1000000$ customers, but only one customer at the second reference date, i.e., $n_2=1$. For more details, we refer to Section \ref{sc: rate of convergence}, where we focus on the convergence for different combinations of $N$ and numbers of customers $n_t$ per reference date $\RD_t$.

\subsection{Covariance between Default States}  \label{Covariance between Default States}

Pursuant to Paragraph 80 EBA-GL \cite{EBA_GL}, when calculating the long-run default rate, institutions may choose between overlapping time windows of one year and non-overlapping time windows.\ According to Paragraph 78 EBA-GL, the reference dates used should include at least all quarterly reference dates. The overlap of the time windows of the one-year default rates leads to correlations between them, which affects the variance of the long-run default rate and thus the acceptance ranges of the test.

In order to be able to describe the distribution of the test statistic as precisely as possible, in Section \ref{Covariance between Default Rates} below, we thus focus on the analysis of the covariances between the default rates, and start by considering the covariance between two default states of individual obligors for overlapping observation periods in this subsection.\ To that end, we consider two reference dates $\RD_{t}$ and $\RD_{s}$ with $\RD_{t}<\RD_{s}<\RD_{t}+1$ and an obligor who can be observed in the period $[\RD_{t},\RD_{s}+1)$. The default state in the period $[\RD_{t},\RD_{t}+1)$ is described by a random variable $x_t\in\left\{0,1\right\}$, the default state in the period $[\RD_{s},\RD_{s}+1)$ by a random variable $x_s\in\left \{0,1\right\}$. That is, a realization of $x_t$ or $x_s$ answers the question whether or not the obligor under consideration defaulted in time period $[\RD_{t},\RD_{t}+1)$ or $[\RD_{s},\RD_{s}+1)$, respectively.

 In order to compute the covariance between the two default states $x_t$ and $x_s$, a detailed knowledge of the time of default is necessary, where, in the case of more than one default event, we focus on the time of the temporally first default. The time of first default is described by a random variable $T_1\in[\RD_{t},\infty)$, which creates a link between default status and time of default, see, e.g., \cite{European Commission}.
 While $T_1$ describes the timing of the first default starting at time $\RD_{t}$, we model the timing of the first default after the default described by $T_1$ and after $\RD_{s}$ using a random variable $T_2^*\in[\RD_{s},\infty)$, with $T_2^*>T_1$. We then define
\[
T_2:=\begin{cases}
T_1,& \text{if }T_1\geq \RD_{s} \text{ and } T_2^* > \RD_{s}+1,\\
T_2^*, & \text{otherwise.}
\end{cases}
\]
If besides $T_1$ more than one default is observed in the interval $[\RD_{s},\RD_{s}+1)$, $T_2$ describes the time of the second default, otherwise $T_2$ is used to model the time of the first default starting from $\RD_s$. We now divide the considered time period into three disjoint intervals 
$$I_1:=[\RD_{t},\RD_{s}], \quad I_2:=(\RD_{s},\RD_{t}+1],\quad \text{and}\quad I_3:=(\RD_{t}+1,\RD_{s}+1],$$
\begin{figure}[htb]
\includegraphics[width=\textwidth]{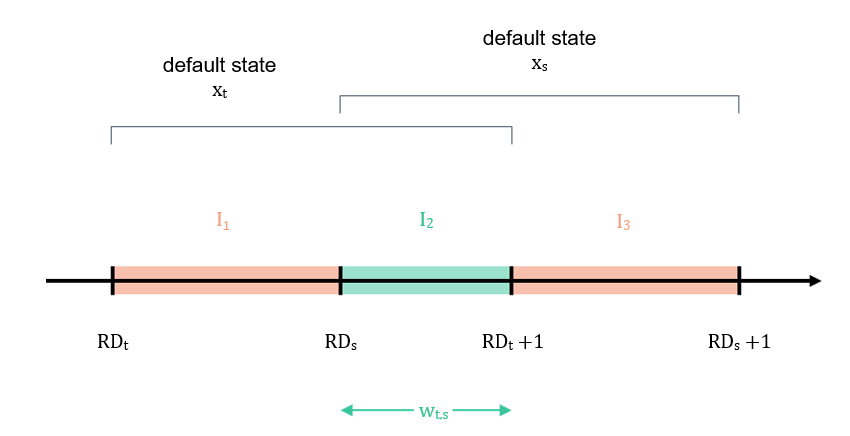}
\caption{Default states per time period}
\label{Overlapping Time Periods}
\end{figure}
and look at the following disjoint events 
\[
   E_1:=\{T_1\in I_1\},\quad   E_2:=\{T_1\in I_2\},\quad\text{and}\quad  E_3:=\{T_1\notin I_1\cup I_2\}.
\]
Moreover, we define
\[
   E_4:=\{T_2\in I_2\},\quad
   E_5:=\{T_2\in I_3\},\quad\text{and}\quad
   E_6:=\{T_2\notin I_2\cup I_3\}.
\]
In order to simplify notation, we set $p_t:=\P(x_t=1)$ and $p_s:=\P(x_s=1)$ and assume that
\[
   \P(x_t=1\:|\:E_5)=p_t\quad\text{and}\quad
   \P(x_s=1\:|\:E_1)=p_s.
\]
In a first step, we focus on the probability $\P(E_1)$. For this, we assume that the probability of default in a certain time interval depends on the length of this interval and on the general creditworthiness of the obligor, i.e., we assume that 
$$\P(T_1\in I\:|\:x_t=1)=f_1(|I|)$$
for every interval $I\subseteq [\RD_{t},\RD_{t}+1)$ and a non-negative and non-decreasing function $f_1$ with $f_1(0)=0$ and $f_1(1)=1$.\ This ensures that the probability of observing a default in a given time interval increases if the length of the interval increases and that the probabilities of default are identical for intervals of the same length. We deduce $$\P(E_1\:|\:x_t=1)=1-\P(E_2\:|\: x_t=1)=1-f_1(w_{t,s}),$$ and therefore, using Bayes' theorem,
\[
1-f_1(w_{t,s})=\P(E_1\:|\:x_t=1)=\P(x_t=1\:|\:E_1)\cdot\frac{\P(E_1)}{\P(x_t=1)}=\frac{\P(E_1)}{p_t}.
\]
Rearranging terms, we find that
$$\P(E_1)=\big(1-f_1(w_{t,s})\big)p_t$$
and, analogously,
\[
   \P(E_2)=f_1(w_{t,s})p_t.
\]
Hence,
\begin{align*}
\E(x_t\cdot x_s)&=\P(x_t=1, x_s =1)=\P(x_t=1,x_s=1\:|\:E_1)\cdot \P(E_1)\\
&\quad+\P(x_t=1,x_s=1\:|\:E_2)\cdot \P(E_2)+\P(x_t=1,x_s=1\:|\:E_3)\cdot \P(E_3)\\
&=\P(x_s=1\:|\:E_1)\cdot \P(E_1)+\P(E_2)= p_s\big(1-f_1(w_{t,s})\big)p_t+f_1(w_{t,s})p_t\\
&=p_tp_s+f_1(w_{t,s})p_t(1-p_s).
\end{align*}
Analogously, one obtains
\[
   \P(E_4)=f_2(w_{t,s})p_s\quad\text{and}\quad
   \P(E_5)=\big(1-f_2(w_{t,s})\big)p_s,
\]
for a non-negative and non-decreasing function $f_2$ with $f_2(0)=0$ and $f_2(1)=1$, and we find that
\begin{align*}
\E(x_t\cdot x_s)&=\P(x_t=1, x_s =1)=\P(x_t=1,x_s=1\:|\:E_4)\cdot \P(E_4)\\
&\quad+\P(x_t=1,x_s=1\:|\:E_5)\cdot \P(E_5)+\P(x_t=1,x_s=1\:|\:E_6)\cdot \P(E_6)\\
&=\P(E_4)+\P(x_t=1\:|\:E_5)\cdot \P(E_5)=f_2(w_{t,s})p_s+p_t\big(1-f_2(w_{t,s})\big)p_s\\
&=p_tp_s+f_2(w_{t,s})p_s(1-p_t).
\end{align*}
Equating the previous expectations results in
\[
f_1(w_{t,s})=f_2(w_{t,s})\frac{p_s\left(1-p_t\right)}{p_t\left(1-p_s\right)},
\]
i.e., $f_1$ is a linear transform of $f_2$.\ To determine the covariance between $x_t$ and $x_s$, knowledge of at least one of the functions $f_1$ or $f_2$ is required.\ We assume that the distribution of the time of default within an observation year is uniform.\ We postulate this assumption for $x_s$, since it gives preference to the more recent information at $\RD_{s}$ over the information at $\RD_{t}$, which lies further in the past.\ We therefore set $f_2(w_{t,s})=w_{t,s}$ and end up with 
\begin{equation} \label{eq:3}
\Cov(x_t,x_s)=\E(x_t\cdot x_s)-\E(x_t)\cdot\E(x_s)=w_{t,s}p_s(1-p_t).
\end{equation}

\subsection{Covariance between Default Rates} \label{Covariance between Default Rates}

Having previously examined the covariance between obligor default states, we extrapolate this result to the covariance of default rates. To do this, we consider a sample of debtors on reference dates $\RD_{t}$ and $\RD_{s}$ with $\RD_{s}<\RD_{t}+1$ and $t,s=1,\ldots,N$. During the transition from the first to the second reference date, debtors can be removed from monitoring due to defaults, terminated business relationships, or migrations to other rating systems, and new debtors can be added on the second reporting date due to new business or migrations to the rating system under consideration.\ In addition, there are debtors who can be observed on both reference dates.\ The number of these so-called persisting customers with respect to reference dates $\RD_{t}$ and $\RD_{s}$ is denoted by $\left|\Lambda_{t} \cap \Lambda_{s}\right|=:k_{t,s}\in \mathbb {N}$. The one-year default rates on the two reference dates are then given by $X_{t}=\frac{1}{n_{t}}\sum_{j\in \Lambda_{t}}^{}x_{t,j}$ with $x_{t,j}\sim B\left(1,p_{t,j}\right)$ and $X_{s}=\frac{1}{n_{s}}\sum_{j\in \Lambda_{s}}^{}x_{s,j}$ with $x_{s,j}\sim B\left(1,p_{s,j}\right)$. Hence,
$$\Cov\left(X_{t},X_{s}\right)=\frac{1}{n_{t} n_{s}}\sum_{i\in \Lambda_{t}}^{}\sum_{j\in \Lambda_{t}}^{}{\Cov\left(x_{t,i},x_{s,j} \right)}.$$
Again, assuming that default states for pairs of different customers are independent, using Formula \eqref{eq:3}, we find that
\begin{equation} \label{eq:Cov(X,Y)}
\Cov\left(X_{t},X_{s}\right)=\frac{1}{n_{t} n_{s}}\sum_{j\in \Lambda_{t} \cap \Lambda_{s} }^{}\Cov\left(x_{t,j},x_{s,j} \right)= \frac{w_{t,s}}{n_{t} n_{s}}\sum_{j\in \Lambda_{t} \cap \Lambda_{s} }^{}p_{s,j}\left(1-p_{t,j} \right),
\end{equation}
where $w_{t,s}$ describes, as before, the size of the overlap between the observation periods starting from the reference dates $\RD_{t}$ and $\RD_{s}$.

\subsection{Variance of the Long-Run Default Rate} \label{Distribution of the long-run default rate}

Based on the considerations for estimating the covariance of default rates in overlapping periods, we now discuss the variance of the long-run default rate. For the random variable $X_t$ with $t=1,\ldots,N$, we define $\E(X_t)=:\mu_t$ and $\sigma(X_t)=:\sigma_t$.\ For the long-run default rate $Z$, which is approximately normally distributed with $Z\sim \mathcal{N}(\mu,\sigma^2)$, we get
$$\mu=\frac{1}{R(N)}\cdot\sum_{t=1}^{N}\mu_t=\frac{1}{R(N)}\cdot \sum_{t=1}^{N}\frac{1}{n_t}\sum_{j \in \Lambda_{t}}^{}p_{t,j}$$
and
\begin{align} \label{est covariance}
R(N)^2\cdot \sigma^2&=\sum_{t=1}^{N}{\sigma_t^2+2\sum_{t=1}^{N-1}\sum_{s=t+1} ^{N}\Cov\left(X_t,X_s\right)} \notag\\
&=\sum _{t=1}^{N}{\sigma _i^2}+2\sum_{i=1}^{q-1}\sum _{t=1}^{N-i}\Cov\left (X_t,X_{t+i}\right),
\end{align}
where the second term contains the covariances caused by overlapping time periods. 
Using Equation \eqref{eq:Cov(X,Y)}, we end up with
\begin{align} \label{Optimization}
R(N)^2\cdot\sigma^2=&\sum_{t=1}^{N}\frac{1}{n_t^2}\sum_{j\in \Lambda_{t}}^{}p_{t,j}(1-p_{t,j})\notag\\
&\quad+\sum_{i=1}^{q-1}\frac{2(q-i)}{q}\sum_{t=1}^{N-i}\frac{1}{n_t \cdot n_{t+i}}\sum_{j\in \Lambda_{t}\cap \Lambda_{t+i}}p_{t+i,j}(1-p_{t,j}). 
\end{align}
For the final calibration test, it will be crucial to estimate the variance of the long-run default rate by a term that solely depends on the expected value $\mu$, the number of debtors per reference date, and debtor-independent variables, see Section \ref{Statistical Test on Portfolio Level}.

\section{Hypothesis test for long-term calibration}\label{hypothesis test for long-term calibration}

The aim of this section is to formulate a statistical test for comparing the realized long-run default rate with the estimated long-run default rate. 
\subsection{Statistical Test per Rating Grade} \label{hypothesis test rating grade}
We start with the simplest case, and consider a portfolio consisting of only one rating grade with an associated probability of default, which we call $\PD_{\class}$. That is, all considered obligors have the same unknown probability of default, which was estimated by $\PD_{\class}$.\ We recall that the measured long-run default rate, denoted by $\LTDR$, is the realization of a random variable $Z$ for which approximately $Z\sim \mathcal{N}(\mu,\sigma^2)$ holds. We consider the hypothesis test, described in Section \ref{Hypothesis Test}, with
\begin{align*}
 &\text{null hypothesis }H_0:\mu=\PD_{\class}\quad\text{and}\\
 &\text{alternative hypothesis }H_1:\mu\neq \PD_{\class}.
\end{align*}
From equation \eqref{Optimization}, we get
\begin{align*}
R(N)^2\cdot\sigma^2=\sum_{t=1}^{N}{\sigma_t}^2+\mu\left(1- \mu\right)\sum_{i=1}^{q-1}\lambda_i,
\end{align*}
where
\[
\lambda_i:=\frac{2(q-i)}{q}\sum_{t=1}^{N-i}\frac{k_{t,t+i}}{n_t\cdot n_{t+i}}.
\]
Moreover,
$$\sum _{t=1}^{N}{\sigma _t}^2=\sum _{t\in R_N}^{}{\frac{1}{n_t}\cdot\mu\left(1- \mu\right)=\mu\left(1-\mu\right)\sum_{t\in R_N}^{}\frac{1}{ n_t}},$$
which implies that
$$\sigma^2=\frac{\mu\left(1- \mu\right)}{R(N)^2}  \left(\sum_{t\in R_N}^{}\frac{1}{n_t}+\sum_{i=1}^{q-1}\lambda_i\right).$$
Therefore, assuming the validity of the null hypothesis, $\sigma$ is known.
In order to define the limits of the acceptance range for a significance level $\alpha\in\left(0,1\right)$,  we choose $k,K\in[0,1]$ in such a way that $$ \mathbb{E}\big( \varphi(Z)\big)\big|_{H_0}\leq\alpha,$$
where $\varphi$ is given by \eqref{eq:phi}. To that end, let
$$k=\PD_{\class}+\Phi ^{-1}\left(\frac{1}{2}\alpha\right)\cdot\sqrt{\frac{ \PD_{\class}\left(1-  \PD_{\class}\right)}{R(N)^2}  \left(\sum_{t\in R_N}^{}\frac{1}{n_t}+\sum_{i=1}^{q-1}\lambda_i\right)}$$
and
$$K=\PD_{\class}+\Phi ^{-1}\left(1-\frac{1}{2}\alpha\right)\cdot\sqrt{\frac{ \PD_{\class}\left(1-  \PD_{\class}\right)}{R(N)^2}  \left(\sum_{t\in R_N}^{}\frac{1}{n_t}+\sum_{i=1}^{q-1}\lambda_i\right)}.$$ 
Then, the calibration test passes if
$$\LTDR\in\left[k,K\right].$$
To see the influence of persisting customers on the acceptance range and the density function of the long-run default rate, we refer to Section \ref{sc. variance} and Section \ref{ex: persisting}.

\subsection{Statistical Test on Portfolio Level} \label{Statistical Test on Portfolio Level}
We now consider a portfolio of heterogeneous obligors, i.e., each obligor has its own individual probability of default.\ For this purpose, we denote the estimated PD of customer $j$ on reference date $\RD_t$ by $\widehat{p}_{t,j}$, the mean PD based on the rating model examined on the reference date $\RD_t$ by $\widehat{\PD}_t$ for $t=1,\ldots,N$, and the long-run central tendency by
$$\LTCT=\frac{1}{N}\cdot \sum_{t=1}^{N}\widehat{\PD}_t=\frac{1}{N}\cdot \sum_{t=1}^{N}\frac{1}{n_t}\sum_{j\in \Lambda_{t}}^{}\widehat{p}_{t,j}.$$
While on the level of rating grades, it is not an exception that the number of customers at a reference date $\RD_t$ might be zero, i.e., $n_t=0$, on a portfolio level, this rarely happens in practice. In this section, we therefore restrict our attention to the case, where $n_t>0$ for all $t=1,\ldots, N$ and therefore $R(N)=N$. 

We consider the hypothesis test from Section \ref{Hypothesis Test} with 
\begin{align*}
 &\text{null hypothesis }H_0:\mu=\LTCT\quad\text{and}\\
 &\text{alternative hypothesis }H_1:\mu\neq \LTCT.
\end{align*}
As in the case of individual rating grades, the test statistic depends on $\mu$ and $\sigma$ but, in this case, Equation \eqref{Optimization} can no longer be reduced to a term depending only on $\mu$.\ More precisely, if the correctness of the null hypothesis is assumed, the distribution parameters are only partially determined.\ If, however, we replace $\sigma$ by an expression $\sigma_{\min}(\mu)$ with $\sigma\geq \sigma_{\min}(\mu)$, the null hypothesis implies a concrete distribution for the underlying test statistic. This results in a narrower confidence interval and therefore a greater likelihood of committing a type I error. At the same time, the probability of a type II error is reduced, comparable to lowering the significance level.

The key challenge is to derive an analytic expression for $\sigma_{\min}$ that depends only on PD-independent model parameters, such as $N$, $M$, $n_t$, etc., and satisfies $\sigma_{\min}(\mu)\leq \sigma$ under the null hypothesis for all possible combinations of $p_{t,j}\in [\PD_{\min},\PD_{\max}]$.  Since we aim to minimize the distance between $\sigma_{\min}(\mu)$ and $\sigma$, it is sensible to define $\sigma_{\min}$ via the solution to a minimization problem with cost functional given by the right-hand side of \eqref{Optimization} under the side condition
\begin{equation} \label{eq: constraint}
    \mu=\frac{1}{N}\cdot \sum_{t=1}^{N}\frac{1}{n_t}\sum_{j\in \Lambda_{t}}^{}p_{t,j}\quad \text{and}\quad p_{t,j}\in[\PD_{\min} ,\PD_{\max} ].
\end{equation}
Observe that the right-hand side of \eqref{Optimization} is, in general, neither convex nor concave with respect to the $p_{t,j}$.\ This together with the high dimensionality of the problem, makes the numerical or analytical computation of solutions is rather involved.\ In order to circumvent this issue, we therefore replace the right-hand side of \eqref{Optimization} by a suitable linear cost functional, which turns the minimization into a linear program that can be solved analytically, up to a sort algorithm, and therefore also numerically with great efficiency.

First, note that mixed terms, i.e., products of PDs of the same customer at different reference dates, always appear in the sums
\[
\sum_{j\in \Lambda_{t}\cap \Lambda_{t+i}}\textcolor{orange}{p_{t+i,j}}\left(1-\textcolor{blue}{p_{t,j}}\right)
\]
 caused by the covariances. 
Hence in Formula \eqref{Optimization}, we want to substitute the subtrahend (blue) by an affine linear function with respect to the multiplier (orange) in a conservative way. On a portfolio level, it is appropriate to assume that an obligor's PD remains constant on average over a one year time horizon, i.e., we assume
\begin{equation}\label{eq.simplification}
   \sum_{j\in \Lambda_{t}\cap \Lambda_{t+i}}\textcolor{orange}{p_{t+i,j}}\left(1-\textcolor{blue}{p_{t,j}}\right) \approx \sum_{j\in \Lambda_{t}\cap \Lambda_{t+l}}\textcolor{orange}{p_{t+i,j}}\left(1-\textcolor{orange}{p_{t+i,j}}\right). 
\end{equation}
For most portfolios, this is not only a plausible but also a fairly conservative assumption regarding the portfolio variance since, for a given vector $(x_1,...,x_n)$ with $0<x_i<1$ for $i=1,\ldots,n$, the sum 
\[
\sum_{i=1}^{n}x_i(1-y_i)
\]
is minimized for $x_i=y_i$, under the assumption that, for each $i=1,\ldots, n$, there exists some $j=1,\ldots, n$ with $x_i=y_j$. Certainly, even more conservative assumptions can be made at this point with $\PD_{\max}$ being the most conservative subtrahend possible. The impact of such a choice is discussed in Section \ref{sc. reduction}, where we indicate that the acceptance ranges hardly change when using $\PD_{\max}$ instead of $p_{t+i,j}$.\ We point out that the subsequent discussion applies also to this choice without limitations.

From Equation \eqref{Optimization}, we get
\begin{align*} \label{Optimization qu}
N^2\cdot\sigma^2&=\sum_{t=1}^{N}\frac{1}{n_t^2}\sum_{j\in \Lambda_{t}}^{}p_{t,j}(1-p_{t,j})\\ 
&\quad +\sum_{i=1}^{q-1}\frac{2(q-i)}{q}\sum_{t=1}^{N-i}\frac{1}{n_t \cdot n_{t+i}}\sum_{j\in \Lambda_{t}\cap \Lambda_{t+i}}p_{t+i,j}(1-p_{t+i,j})\\
&=\sum_{t=1}^{N}\frac{1}{n_t^2}\sum_{j\in \Lambda_{t}}g_0(p_{t,j})+2\sum_{i=1}^{q-1}\sum_{t=1}^{N-i}\frac{1}{n_t \cdot n_{t+i}}\sum_{j\in \Lambda_{t}\cap \Lambda_{t+i}}g_i(p_{t+i,j})
\end{align*}
  with
\[
g_{i}\left(x\right):=\frac{q-i}{q}\cdot x \left(1-x \right)\quad\text{for }x\in (0,1)\text{ and }i=0,\ldots,q-1.
\]
Now, we replace each of the functions $g_i$ by an affine linear function $f_{i}$ such that the functions $g_{i}$ and $f_i$ coincide on $\PD_{\max}$ and $\PD_{\min}$. Since $g_i$ is concave and $f_i$ is affine linear, this implies that $g_i(x)\geq f_i(x)$ for all $x\in [\PD_{\min},\PD_{\max}]$.\ To be precise, the functions $f_i$ are given by 
\[
f_{i}(x):=\alpha_i\cdot x+c_i \quad\text{for all }x\in (0,1)
\]
with
\begin{align*}
 \alpha_i:=\frac{q-i}{q}(1-\PD_{\max}-\PD_{\min})\quad \text{and}\quad c_i=\frac{q-i}{q}\PD_{\min}\PD_{\max}
\end{align*}
for all $i=0,\ldots,q-1$. We therefore end up with the estimate
\begin{align} 
N^2\cdot \sigma^2 &= \sum_{t=1}^{N}\frac{1}{n_t^2}\sum_{j\in \Lambda_{t}}g_0(p_{t,j})+2\sum_{i=1}^{q-1}\sum_{t=1}^{N-i}\frac{1}{n_t \cdot n_{t+i}}\sum_{j\in \Lambda_{t}\cap \Lambda_{t+i}}g_i(p_{t+i,j})\notag \\
&\geq \sum_{t=1}^{N}\frac{1}{n_t^2}\sum_{j\in \Lambda_{t}}f_0(p_{t,j})+2\sum_{i=1}^{q-1}\sum_{t=1}^{N-i}\frac{1}{n_t \cdot n_{t+i}}\sum_{j\in \Lambda_{t}\cap \Lambda_{t+i}}f_i(p_{t+i,j})\notag \\
&=\sum_{t=1}^{N}\sum_{j\in \Lambda_{t}}^{}\alpha_{t,j}p_{t,j}+C, \label{Optimization lin}
\end{align}
where
\begin{align*}
 C=c_0\cdot \sum_{t=1}^{N}\frac{1}{n_t}
+2\cdot \sum_{i=1}^{q-1} c_{i} \cdot \sum_{t=1}^{N-i}\frac{ k_{t,t+i}}{n_t \cdot n_{t+i}}
\end{align*}
 and
\begin{align*}
 \alpha_{t,j}=\frac{1}{n_t^2}\alpha_0 \eins_{\Lambda_{t}}(j)+\frac{2}{n_t}\cdot \sum_{i=1}^{q-1} \frac{\alpha_{i}}{n_{t-i}}\eins_{\Lambda_{t-i}\cap \Lambda_{t}}(j)
\end{align*}
with $\Lambda_{t-i}:=\emptyset$ if $t-i<1$ for $t=1,\ldots, N$ and $j=1,\ldots, M$.\ We have therefore reduced the original problem to
\begin{equation} \label{eq: reduced problem}
   \text{minimize}\quad \sum_{t=1}^{N}\sum_{j\in \Lambda_{t}}^{}\alpha_{t,j}p_{t,j}\quad\text{under the side condition \eqref{eq: constraint}.}
\end{equation}
As shown in the Appendix \ref{ Minimization problem}, the solution $p(\mu)$ to the minimization problem \eqref{eq: reduced problem} can be easily computed analytically and numerically. We denote the minimal value by $m(\mu)$, and define $\sigma_{\min}^{2}(\mu)=\frac{m(\mu)+C}{N^2}$.

For the modified random variable $Z^*\sim \mathcal{N} \big(\mu,\sigma_{\min}^2(\mu)\big)$, we are now able to define the limits $k,K\in[0,1]$ of the acceptance range for a significance level $\alpha\in\left(0,1\right)$ in such a way that $ \mathbb{E}\left( \varphi\left(Z^*\right)\right)|_{H_0}\leq\alpha.$ For this we define
$$k:=\LTCT+\Phi^{-1} \left(\frac{1}{2}\alpha \right) \cdot \sigma_{\min}(\LTCT)$$
and
$$K:=\LTCT+\Phi^{-1} \left(1-\frac{1}{2}\alpha \right) \cdot \sigma_{\min}(\LTCT).$$
The calibration test passes if
$$\LTDR\in\left[k,K\right].$$

\section{Discussion and further considerations}\label{sec:4}

\subsection{Effect of Persisting Customers on the Variance of $Z$} \label{sc. variance}
 In this section, we aim to highlight the effects caused by persisting customers on the distribution of $Z$. We calculate the density function of the long-run default rate using a sample portfolio in a rating grade with $\PD=0.02$ and set $q=4$. The number of reference dates is $N=R(N)=32$, i.e., we consider a portfolio with a relatively short history of eight years.\ The number of customers is set to be constant over time with $n=50$.\ We furthermore assume that the ratio of persisting customers per time interval is constant over time.\ We have $k_{t,t+1} =45$ for all $t\leq 31$, $k_{t,t+2} =40$ for all $t\leq 30$ and $k_{t,t+3}=35$ for all $t\leq 29$. 
Only about $70\%$ of the persisting customers are still in the considered rating grade after 3 quarters, which represents a relatively high fluctuation. \\ 
These specifications determine the density function of the long-run default rate (orange).\ In comparison, one can see the density functions of the long-run default rate with a maximum and minimum number of persisting customers in green and blue, respectively. We have 
$$\sigma^2_{\text{orange}}\approx 4.13\cdot 10^{-5},\quad \sigma^2_{\text{blue}}\approx 1.23\cdot 10^{-5},\quad\text{and}\quad \sigma^2_{\text{green}}\approx 4.71\cdot 10^{-5}.$$
Figure \ref{fig:variance} shows that that correlation effects caused by overlapping time windows should not be neglected since, otherwise, the acceptance range would be way too tight, see also Example \ref{ex: persisting} below.
\begin{figure}[htb]\label{fig:variance}
\includegraphics[width=\textwidth]{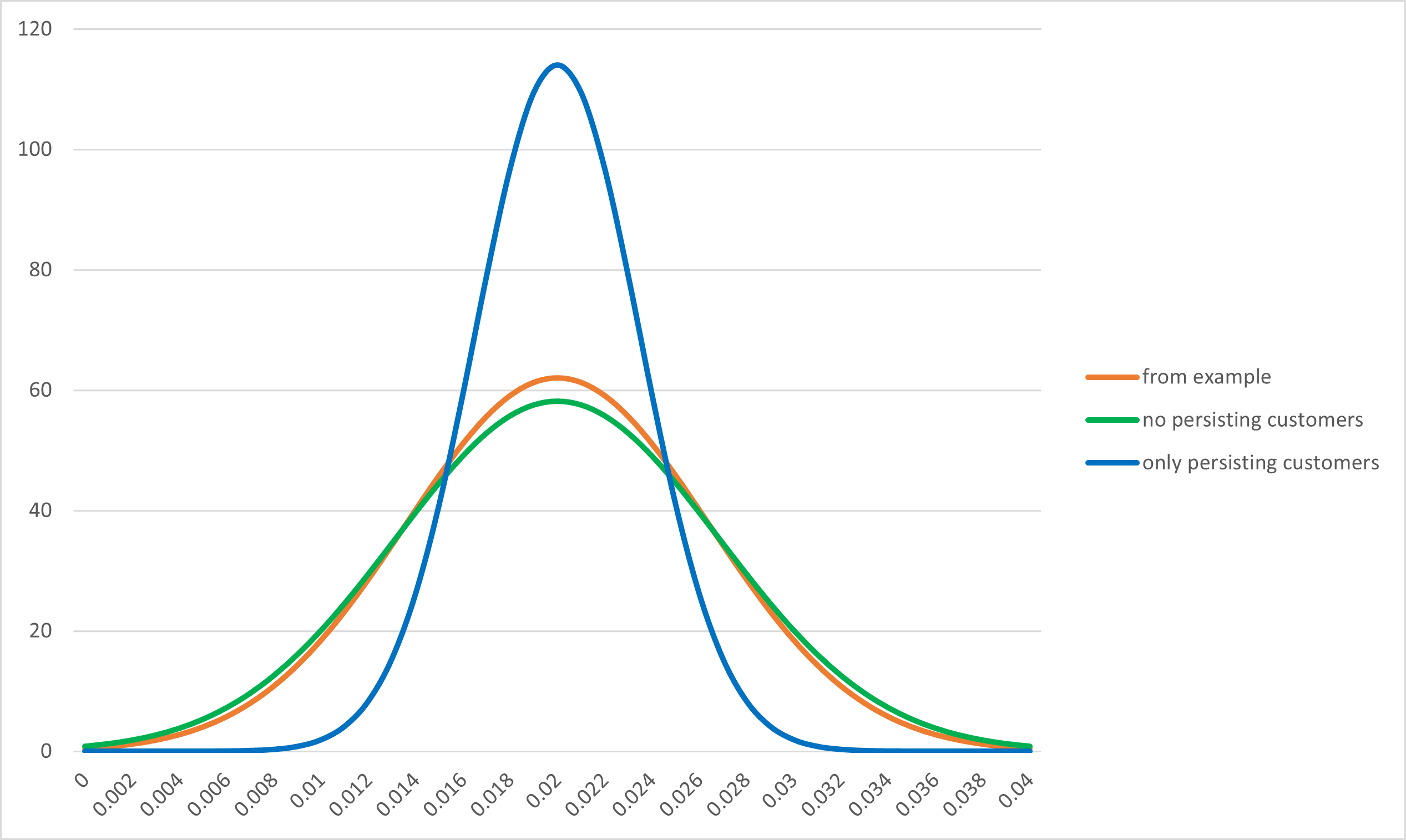}
\caption{Density functions of long-run default rates}
\label{Density Functions}
\end{figure}

\subsection{Effect of Persisting Customers on Acceptance Range} \label{ex: persisting}
We examine the influence of the number of persisting customers on the width of the acceptance range. Again, we choose $q=4$ and, for the sake of simplicity, we consider the case at the level of the individual grades choosing $\PD_{\class} = 0.02$. The proportion of persisting customers is usually lower at the level of individual grades than in the overall portfolio, since a rating migration automatically means that the customer in question is no longer in the sample for the individual grade but is still in the overall portfolio.

The number of reference dates is $N=R(N)=60$ and the number of customers is constant, i.e., $n_t:=n=50$ for $t=1,\ldots,50$. The extreme cases regarding persisting customers now represent the two scenarios in which, on the one hand, all customers remain the same over the entire history and, on the other hand, none of the customers in a specific quarter occurs in the previous quarter or in the following quarter.\\
The first case implies $k_{t,t+1}=n$ for $t\leq N-1$, $k_{t,t+2}=n$ for $t\leq N-2$ and $k_{t,t+3}=n$ for $t\leq N-3$. Thus, for two one-year default rates $X_{t}$ and $X_{s}$ where $1\leq t< s \leq 50$, we have 
$$\Cov\left(X_{t},X_{s} \right)=\frac{1}{n}\cdot w_{t,s} \left( 1- \PD_{\class} \right)\cdot \PD_{\class}.$$
For the second case, we have $k_{t,t+1}=0$ for $t\leq N-1$, $k_{t,t+2}=0$ for $t\leq N-2$ and $k_{t,t+3}=0$ for $t\leq N-3$, implying 
$$\Cov\left(X_{t},X_{s} \right)=0
$$
for two one-year default rates $X_{t}$ and $X_{s}$. Since the covariance between default rates is always $0$, in this case, the distribution of the long-run default rate is the same as if correlation effects due to overlapping time windows were ignored completely. Note that such a scenario is extremely unlikely, especially for a large number of ratings, and is therefore purely hypothetical since
customers would constantly switch rating grades, which would imply a high level of instability in the rating system.\\
In Figure \ref{Range of Bounds}, we can see that the width of the acceptance range to a given level $\alpha$ for a portfolio only consisting of persisting customers is almost doubled compared to a portfolio with no persisting customers with respect to a three months time horizon.
\begin{figure}[htb]
\includegraphics[width=\textwidth]{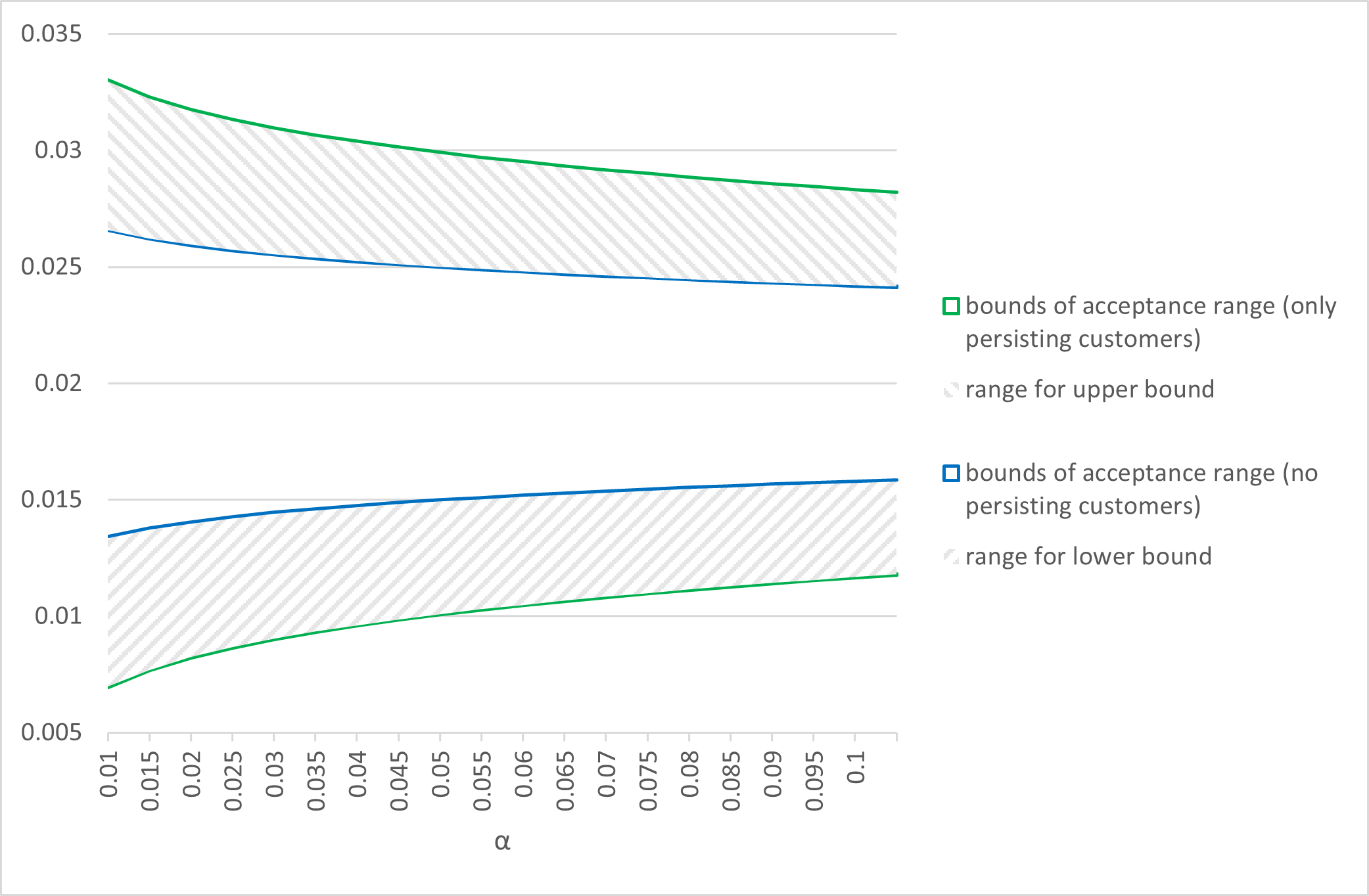}
\caption{Effect of proportion of persisting customers on acceptance range per confidence level $\alpha$}
\label{Range of Bounds}
\end{figure}
\begin{figure}[htb]
\includegraphics[width=\textwidth]{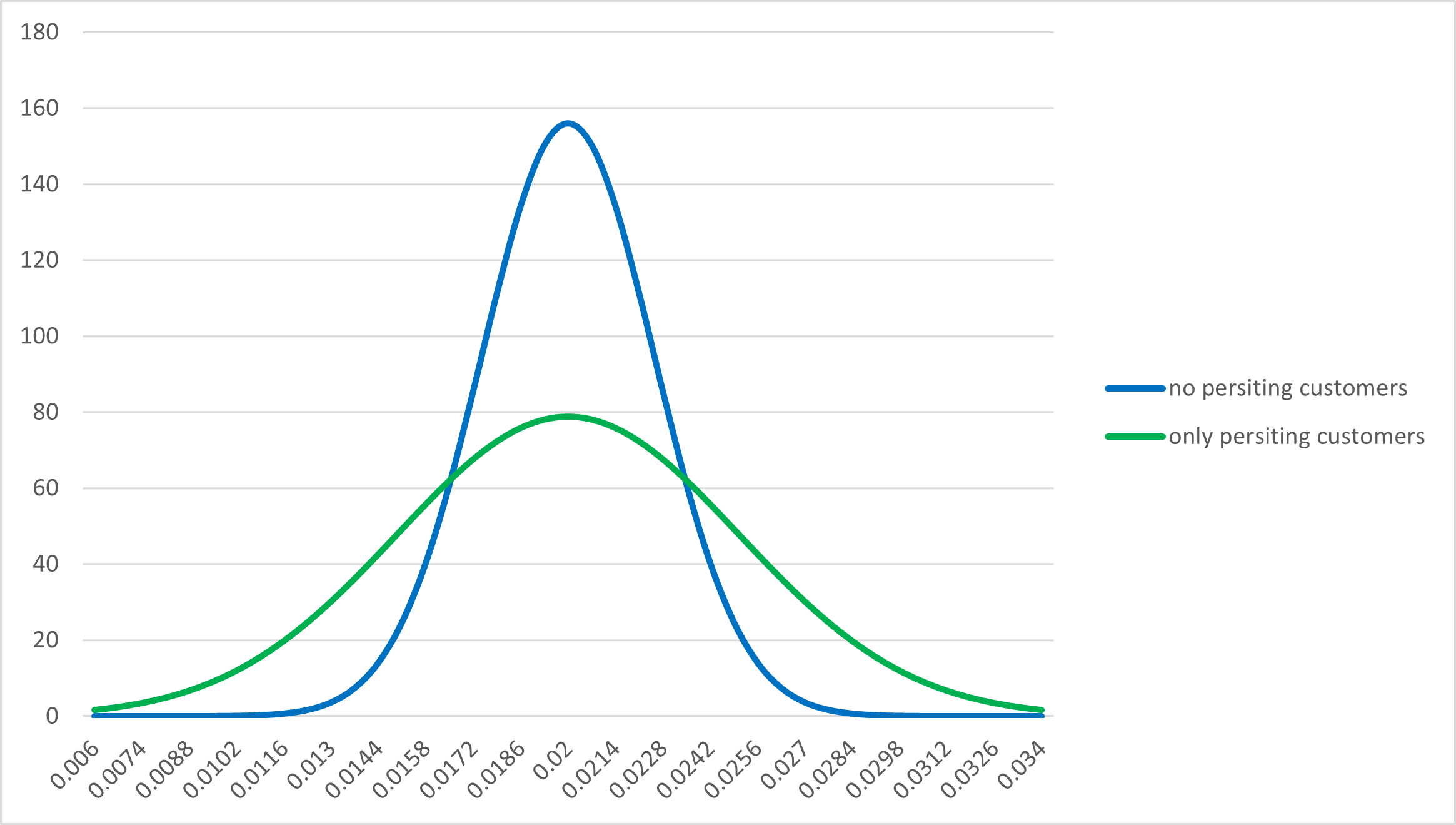}
\caption{Comparison between density functions of long-run default rates}
\label{Densities_2}
\end{figure}

\subsection{Some Thoughts on the Rate of Convergence} \label{sc: rate of convergence}
In view of Section \ref{Normal Distribution of long run default rate}, we briefly discuss the rate of convergence and mention a few particularities that need to be taken into account concerning the asymptotic behaviour of the long-run default rate $Z$.\ When using the classical central limit theorem, it is common to specify a minimum number of random variables ensuring a reasonable approximation of the normal distribution. In our case, additional factors need to be taken into account since, for example, a large number of customers, i.e., a large number $M$ in 
$$
Z=\frac{1}{N}\sum\limits_{t=1}^{N}\frac{1}{n_t}\sum\limits_{j=1}^M y_{t,j}=\sum\limits_{j=1}^M Y_j,
$$
does not directly lead to a good approximation. This stems from the calculation logic of the long-run default rate. For a simple illustration, we assume that we are in the setting, where all customers have the same PD with $\PD=0.01$. For $N=1$ and $M=1000$ or, analogously, $N=1000$ and $n_t=1$ for $t=1,\ldots, 1000$, the normal approximation is way better than for $N=2$ with $n_1=1$ and $n_2=999$ since, in the first case, $\P(Z\geq 0,5)\approx 0$ while, in the second case, $\P(Z\geq 0,5)>0,01.$ Thus, the rate of convergence depends on the number of customers, the number of reference dates, and the number of customers per reference date. While the previously described scenarios are uncommon at portfolio level, at the level of individual rating grades, constellations can arise in which, on certain reference dates, there is only one customer in the corresponding rating grade.

In the sequel, we elaborate more on this particularity and aim to provide a rule of thumb, i.e., conditions for $N$, $n_{\max}$, and $n_{\min}$ that imply a satisfactory normal approximation. To that end, we consider a synthetic portfolio. For the sake of simplicity, we assume that $q=1$ and simulate a convolution of weighted binomial distributions. We study the difference between the distribution functions of the long-run default rate on the level of a rating grade with $\PD=0.01$ when using a simulation on the one hand and using the normal approximation on the other hand in different scenarios. We assume the number of reference dates is $N=56$.\ The number of customers $n_t$ per reference date $\RD_t$ and thus the number of all customers $M$ in the portfolio varies between the considered scenarios.

\begin{figure}[htb]
\includegraphics[width=\textwidth]{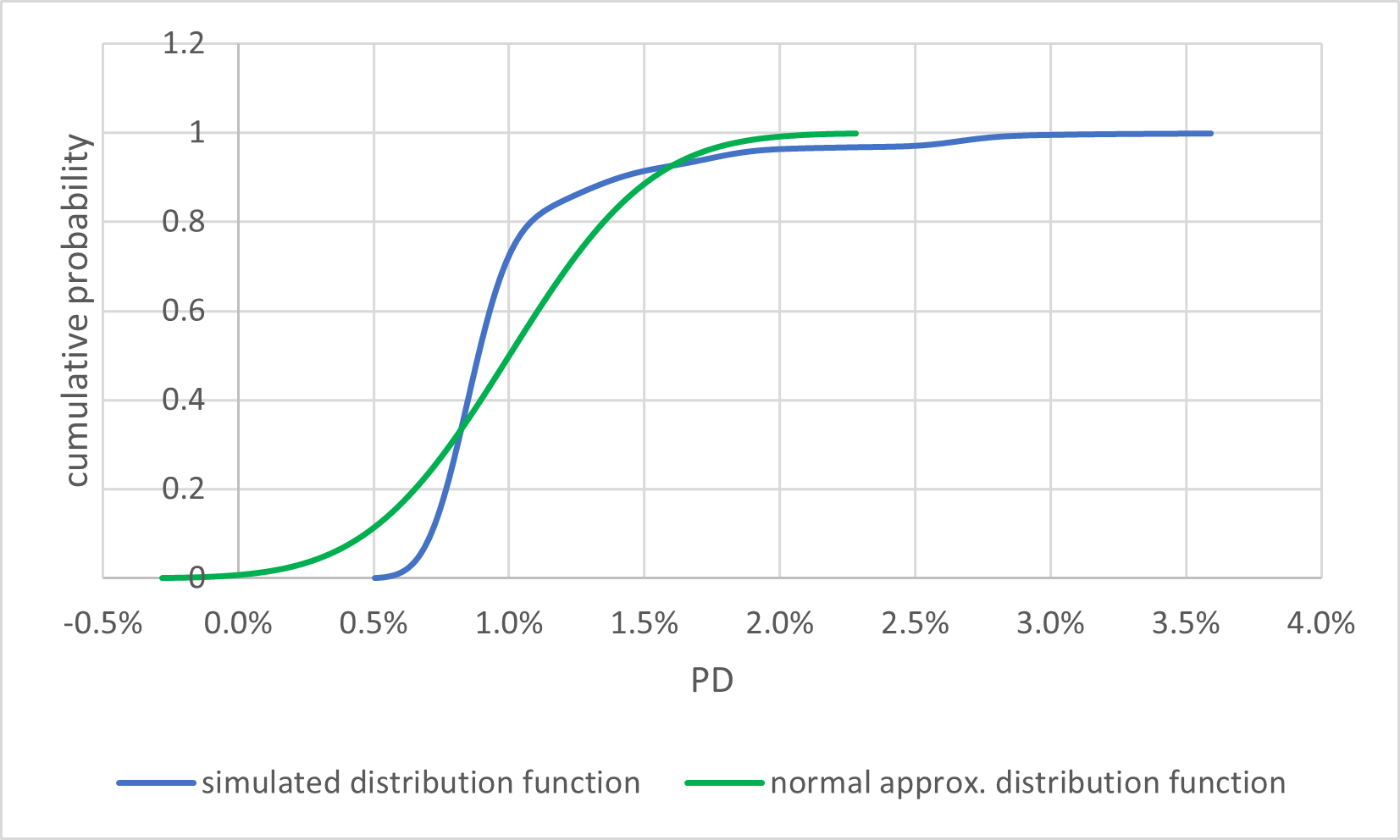}
\caption{Distribution functions:\ simulation and normal approximation with $n_t\leq5$ for $t=1,\ldots,8$, and $n_t=100$ for $t=9,\ldots, 56$  }
\label{distribution functions 1}
\end{figure}

\begin{figure}[htb]
\includegraphics[width=\textwidth]{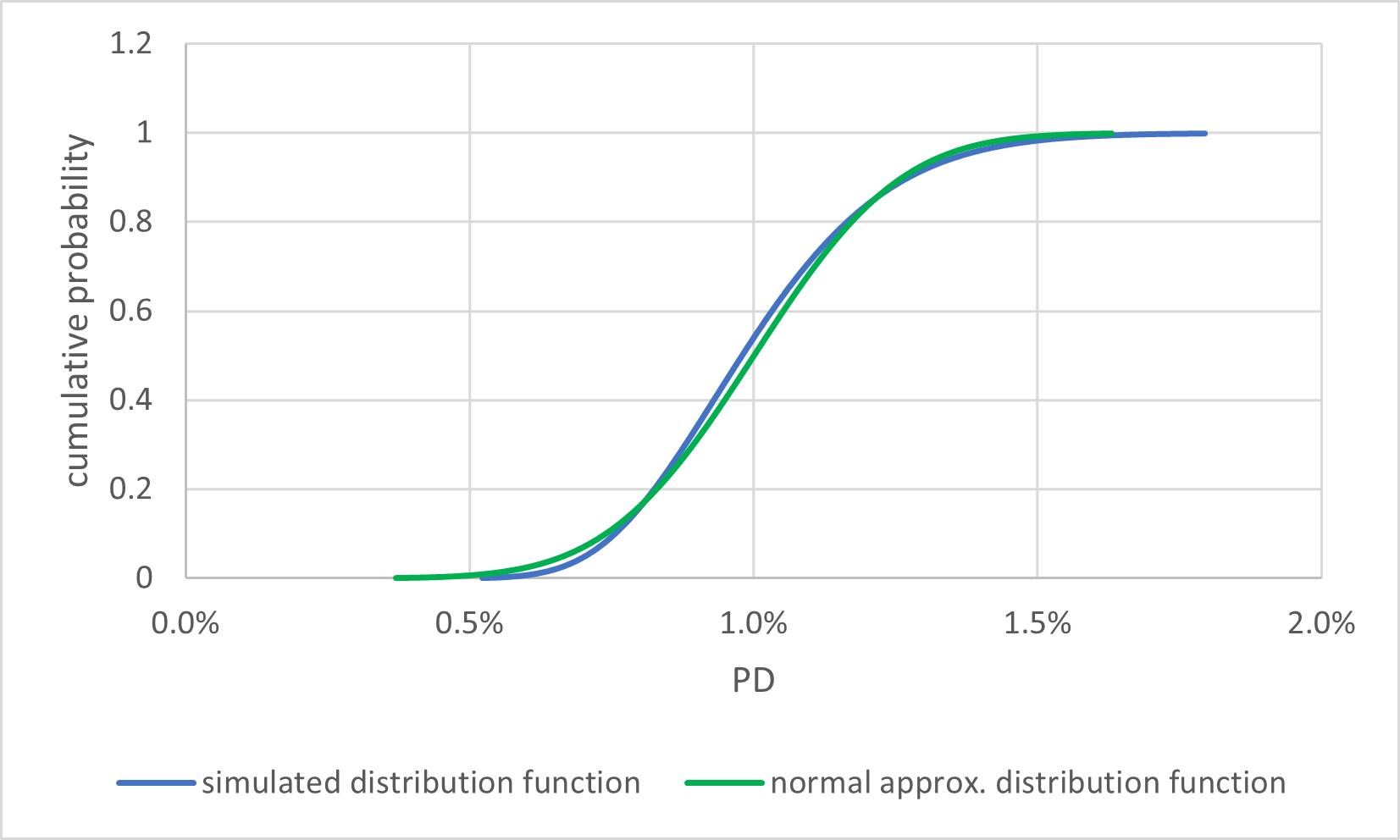}
\caption{Distribution functions:\ simulation and normal approximation with $n_t\leq10$ for $t=1,\ldots,8$, and $n_t=100$ for $t=9,\ldots, 56$  }
\label{distribution functions 2}
\end{figure}

\begin{figure}[htb]
\includegraphics[width=\textwidth]{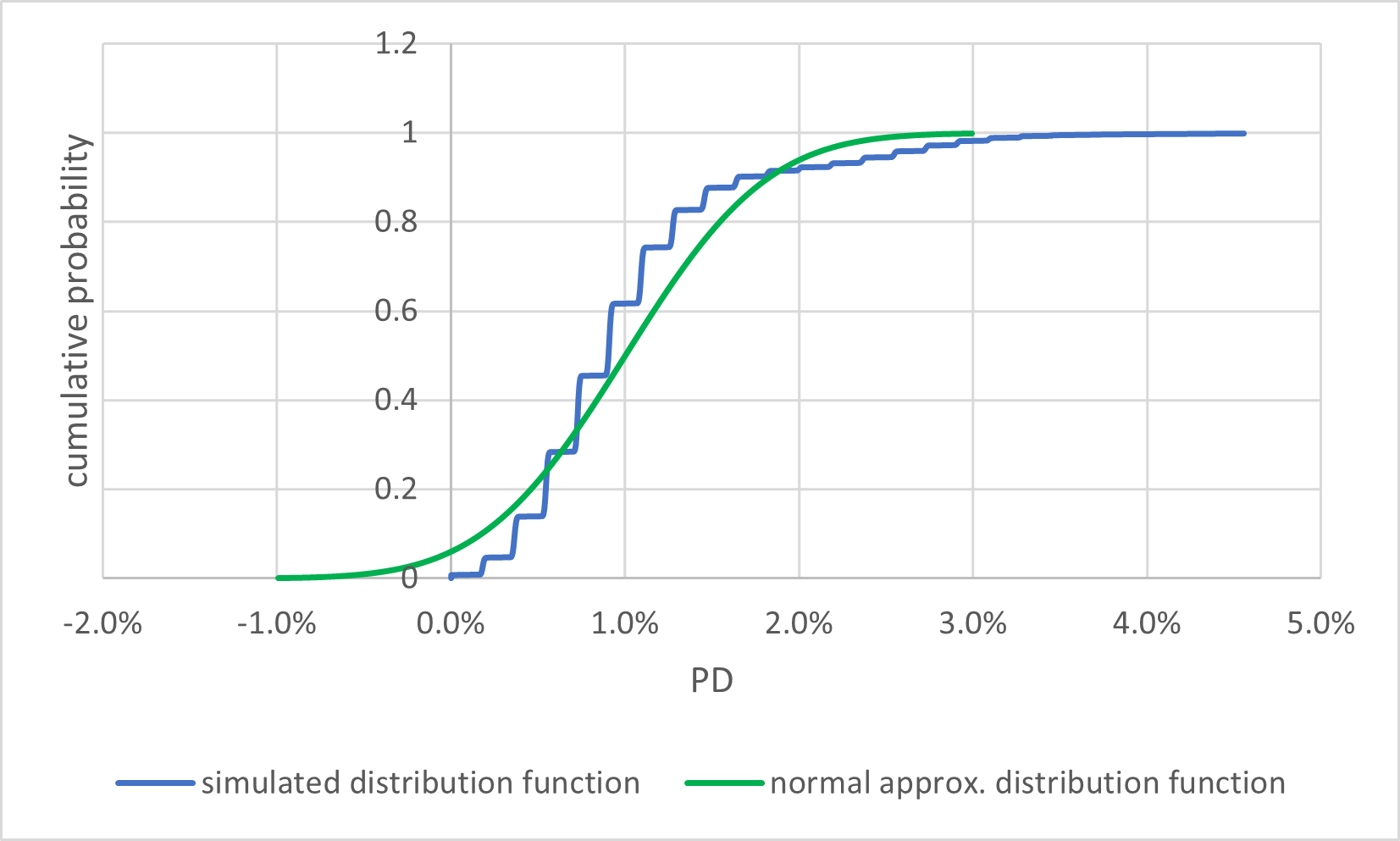}
\caption{Distribution functions:\ simulation and normal approximation with $n_t=1$ for $t=1,\ldots,8$, and $n_t=10$ for $t=9,\ldots, 56$  }
\label{distribution functions 4}
\end{figure} 
\begin{figure}[htb]
\includegraphics[width=\textwidth]{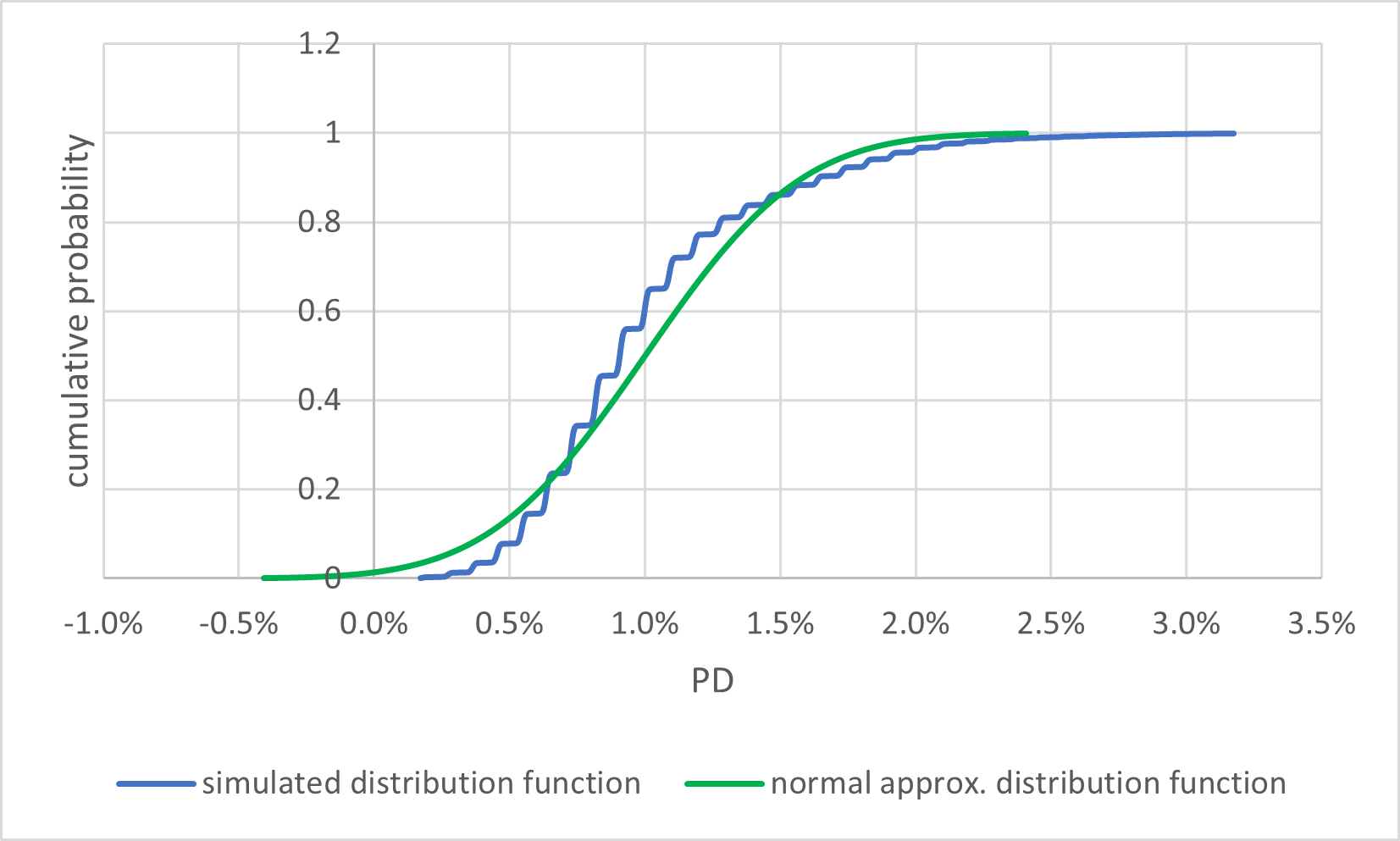}
\caption{Distribution functions:\ simulation and normal approximation with $n_t=2$ for $t=1,\ldots,8$, and $n_t=20$ for $t=9,\ldots, 56$  }
\label{distribution functions 5}
\end{figure} 

In Figure \ref{distribution functions 1}, there are eight reference dates where only one to five customers are in the portfolio. On the other reference dates, we assume the number of customers to be constant with $n_t=100$. We observe large differences in the two distribution functions, especially in the tails, which are particularly relevant for the test.\ If we now increase the number of customers on the critical reference dates (those with $n_t\leq5$) to $10$, we observe, in Figure \ref{distribution functions 2}, that the distribution functions are almost the same. Hence, the poor convergence in Figure \ref{distribution functions 1} is caused by the few reference dates with very few customers. Furthermore, we see that, in this special case, $\frac{n_{\min}}{n_{\max}}\geq \frac{1}{10}$ seems to be an appropriate bound in order to obtain a satisfactory approximation.

For a second comparison, we now reduce the number of customers significantly to $1\leq n_t\leq 10$. Again, the number of customers equals one only on the first eight reference dates. In Figure \ref{distribution functions 4}, we still see major differences in the distribution functions. Doubling the number of customers leads to Figure \ref{distribution functions 5}, where we already obtain useful results.

On the other hand, if the number of customers is very low, e.g., $n_t=2$ for half of all reference dates, we still get a good approximation, cf.\ Figure \ref{distribution functions 6}. 
\begin{figure}[htb]
\includegraphics[width=\textwidth]{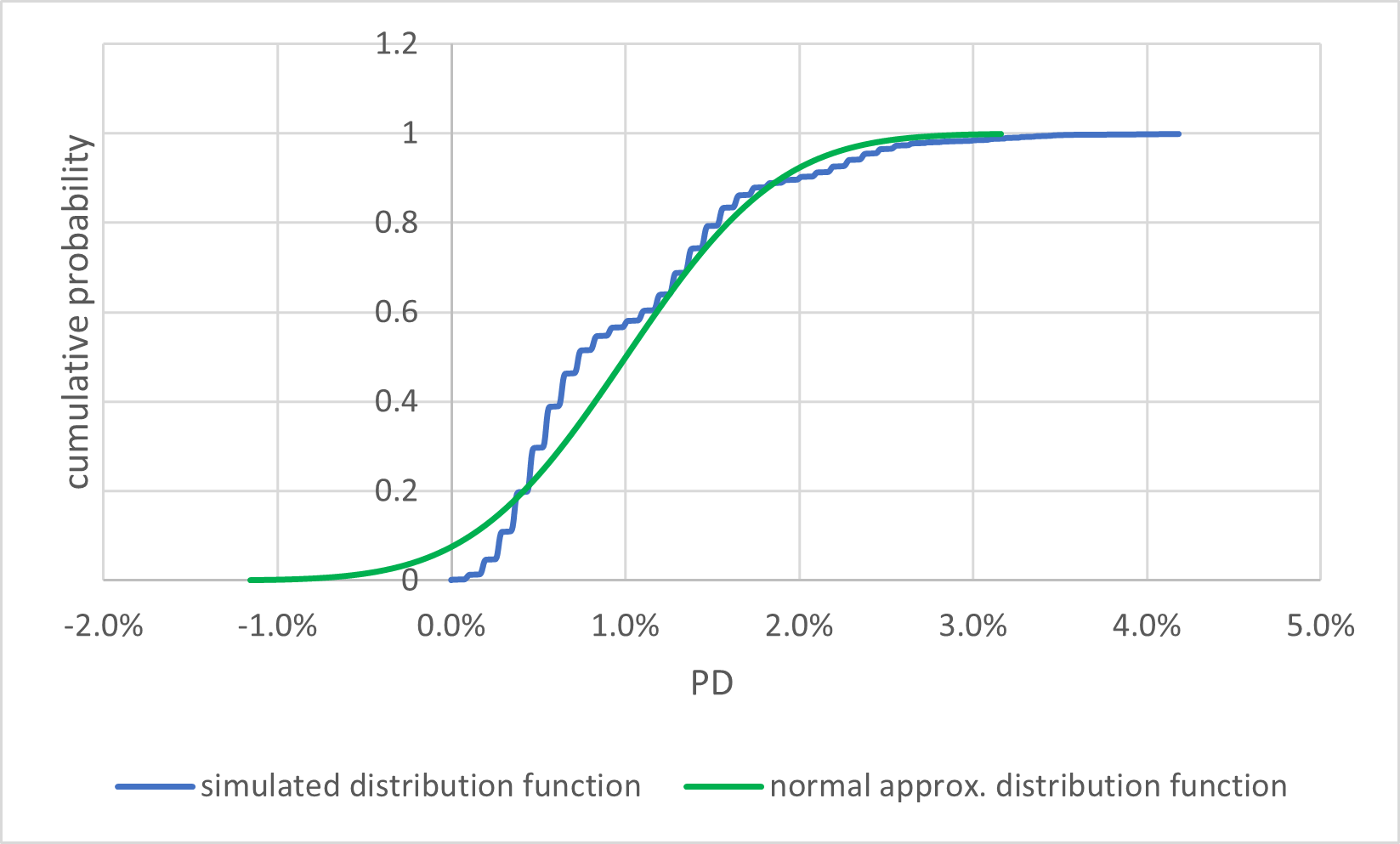}
\caption{Distribution functions:\ simulation and normal approximation with $n_t=2$ for $t=1,\ldots,28$, and $n_t=20$ for $t=29,\ldots, 56$  }
\label{distribution functions 6}
\end{figure} 
From our examples and statements from literature, we conclude that $N\geq30$, $n_{\min}\geq2$, and $\frac{n_{\min}}{n_{\max}}\geq \frac{1}{10}$ seem to be useful conditions to guarantee a satisfying normal approximation.

\subsection{An alternative way to bound the variance}

In the Appendix \ref{Appendix E}, we propose a different way to estimate the variance in an even more conservative way. If one accepts the associated stricter test as a conservative check of the calibration, one can use 
\[
\sigma_{{\rm alt}}^2(\mu):= \frac{1}{N}\left(C+K_1 \right)\mu-\frac{K_2}{N^2}
\]
as the variance of the test statistic with parameters $K_1$ and $K_2$ depending only on the underlying portfolio, see Appendix \ref{Appendix E} for the details.

In this case, we do not have to solve an optimization problem, and can determine the variance directly using the portfolio-specific parameters. However, we additionally need conservative estimates for two parameters that are not directly determined by the null hypothesis. Again, we refer to the  Appendix \ref{Appendix E} for the details.\ We briefly discuss the consequences of this alternative method at this point.

The determination of $\sigma_{{\rm alt}}^2(\mu)$ is based on a successive estimation of the variance of the test statistic against the maximum conservative value in each single step. This procedure can therefore be understood as a maximally conservative solution to the test problem.\ In the case of a portfolio with approximately constant size, approximately constant distribution of default probabilities over time, and no other special features, the conservatism can often be tolerated.\ In case of build-up or run-down portfolios or portfolios with a focus on high investment grade obligors, the risk of a type I error may become disproportionately large and the test may therefore not be suitable.\ This can be easily illustrated by the following two points.\ The estimation of quadratic terms on the left-hand side of \eqref{eq: estimate with C} is done using $\PD_{\max}$ regardless of how many obligors belong to the associated rating grade. For portfolios with mainly obligors with good ratings, this estimate is certainly too conservative. Moreover, the parameters $K_1$ and $K_2$ depend on
\[
\min\bigg\{\frac{1}{n_t}\,\bigg|\,t=1,\ldots, N\text{ and } n_t> 0\bigg\}\quad\text{and}\quad\min\big\{k_{t,t+i}\,\big|\,t=1,\ldots, N-i\big\},
\]
i.e., for portfolios where the number of customers or the number of persisting customers changes strongly over time, a lot of information is lost with this method. Thus, the estimate on the variance of the test statistic in Section \ref{Statistical Test on Portfolio Level} is much less conservative and suitable for significantly more portfolios.

\subsection{Additional Conditions on the Rating Distribution}

In this section, we give an outlook and offer suggestions as to which requirements can be placed on rating distributions. The implementation of some suggestions is briefly outlined. The rating distribution of $p(\mu)$ obtained by minimizing \eqref{eq: reduced problem} tends to be U-shaped, cf. Figure \ref{Distribution 1}, with all but at most one rating in the best or the worst rating grade.\ A U-shaped rating distribution, i.e., a distribution almost exclusively between the two extreme grades is highly unusual, both in the case of a large number of ratings and in a portfolio that hardly has any customers in the lower rating grade.
\begin{figure}[htb]
\includegraphics[width=\textwidth]{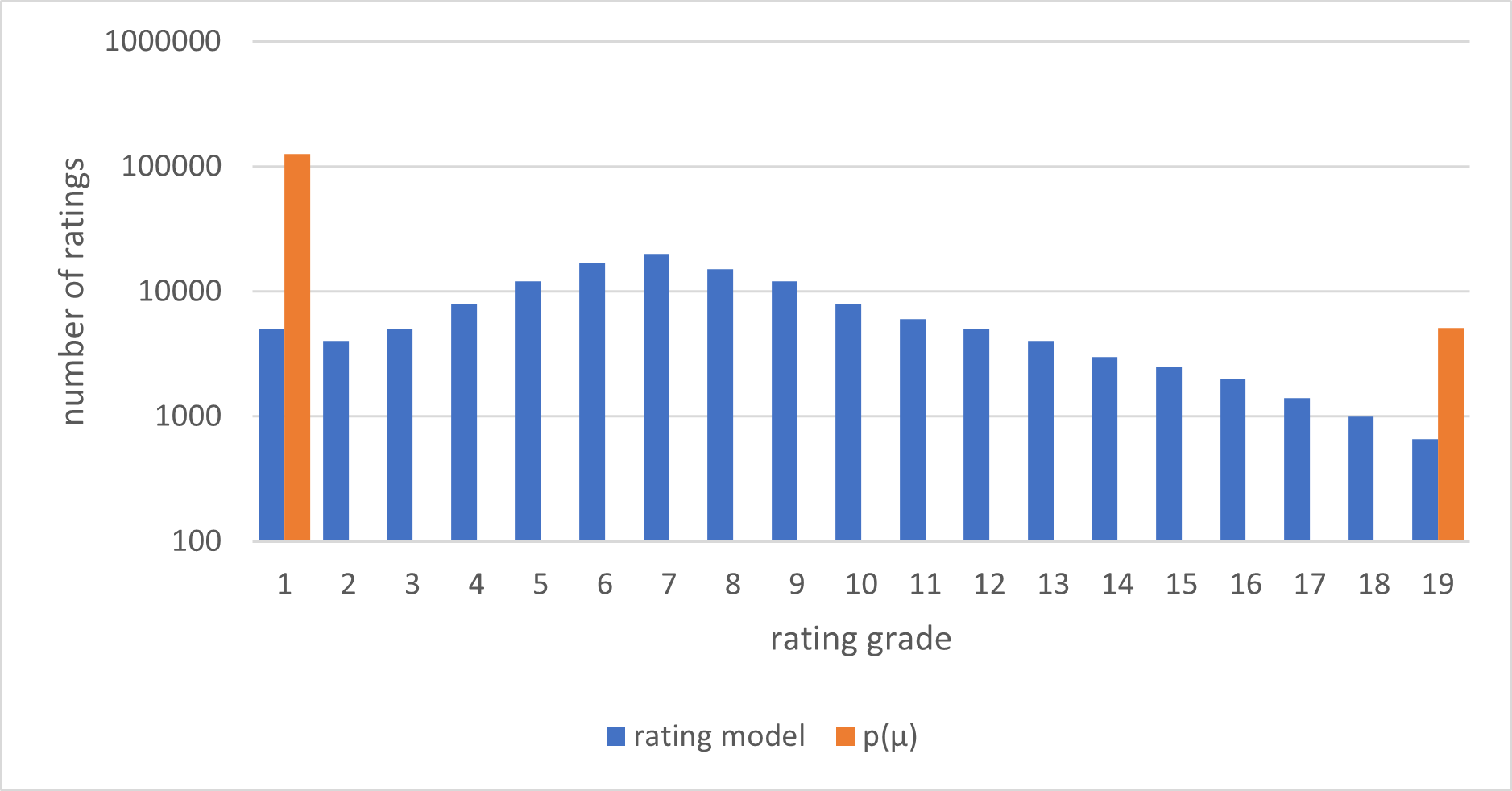}
\caption{Rating distribution of rating model and rating distribution after minimization without additional assumptions on the distribution}
\label{Distribution 1}
\end{figure}

Moreover, high values of $\PD_{\max}$ lead to a narrowing of the acceptance range.\ Depending on the portfolio, this reduction can be disproportionately strong if, for example, the bad rating grades in low-default portfolios are heavily underrepresented.\ Hence, the question may be raised why the rating distribution, estimated by the rating model, is not used directly.\ As already mentioned, this would, however, require assumptions that go far beyond the null hypothesis, and specifying the rating distribution in that way, would automatically determine the null hypothesis.\ However, it may be sensible to place additional yet conservative conditions on the rating distribution in \eqref{eq: constraint}.\ 

For example, one could demand that ratings exhibit a conservative distribution, given by a function $\varphi_1\colon \{\PD_{m-n+1},...,\PD_m \}\to [0,1]$ across the $n$ worst rating grades with
\[
\P\big(p_{t,j}=\PD_k \: \big| \: p_{t,j}\geq \PD_{m-n+1} \big)=\varphi_1(\PD_k) \quad\text{for }k\geq m-n+1.
\]
For most portfolios a conservative choice of $\varphi_1$ would be $\varphi_1\left(\PD_k  \right)=\frac{1}{n}$, cf.\ Figure \ref{Distribution 2}.

Another approach would be to trust the distribution estimated by the rating model to some extent, and require that there is at least a fixed fraction $\delta \in [0,1]$ of the number of ratings per grade given by this distribution with density $\varphi_2\colon\{\PD_{1},...,\PD_m \}\to [0,1]$.\ Then, in addition to \eqref{eq: constraint}, one might demand
\[
n\left(\PD_k \right) \geq n_{\min}\left(\PD_k \right) := \Bigg\lfloor \sum_{t=1}^{N}n_t \cdot \varphi_2(\PD_k) \cdot \delta \Bigg\rfloor,
\]
where $n\left(\PD_k \right)$ is the number of the ratings in rating grade $k$, cf.\ Figure \ref{Distribution 3}.

\begin{figure}[htb]
\includegraphics[width=\textwidth]{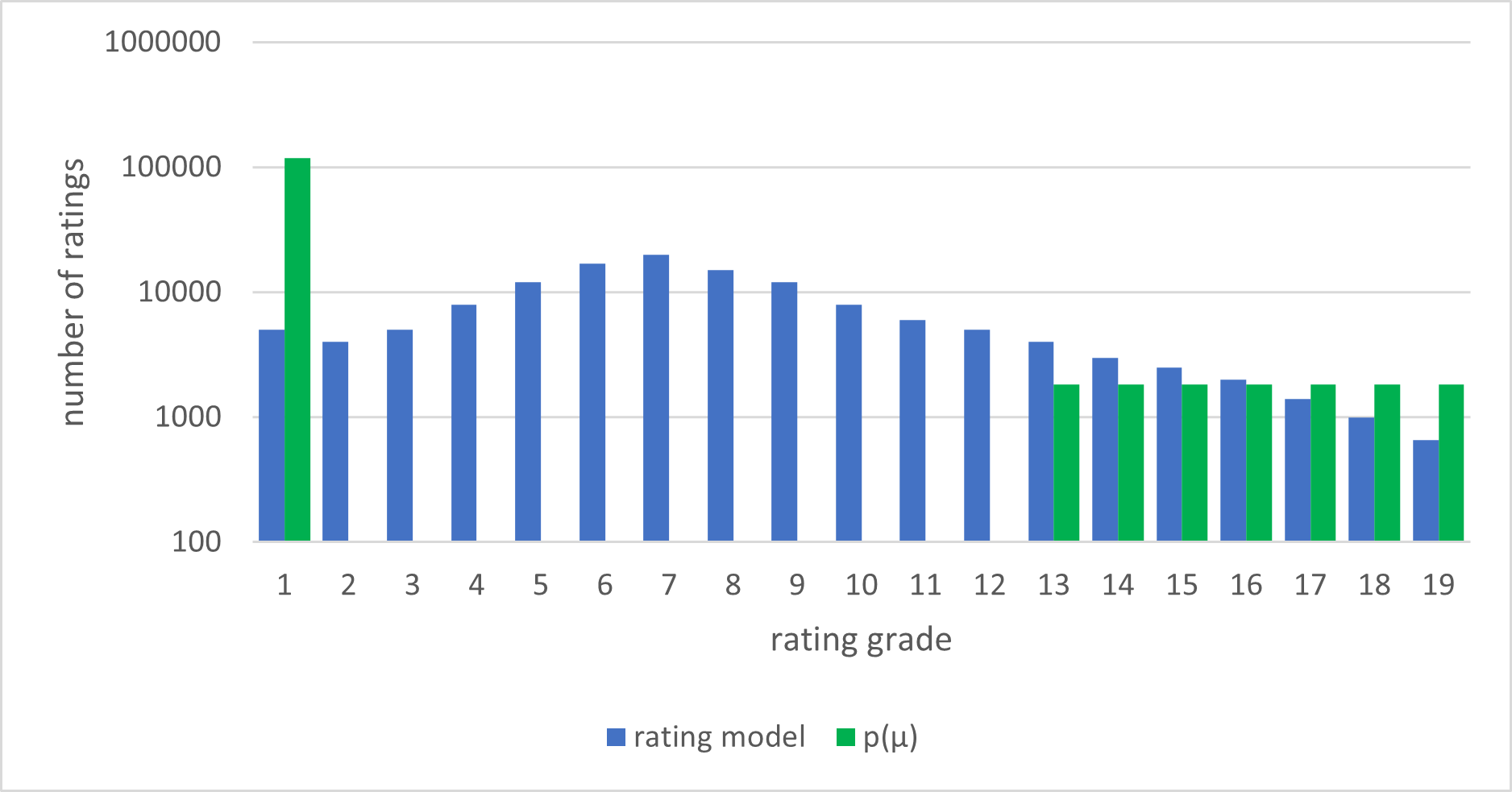}
\caption{Rating distribution of rating model and rating distribution after minimization with assumption of uniform distribution for bad rating grades}
\label{Distribution 2}
\end{figure}
\begin{figure}[htb]
\includegraphics[width=\textwidth]{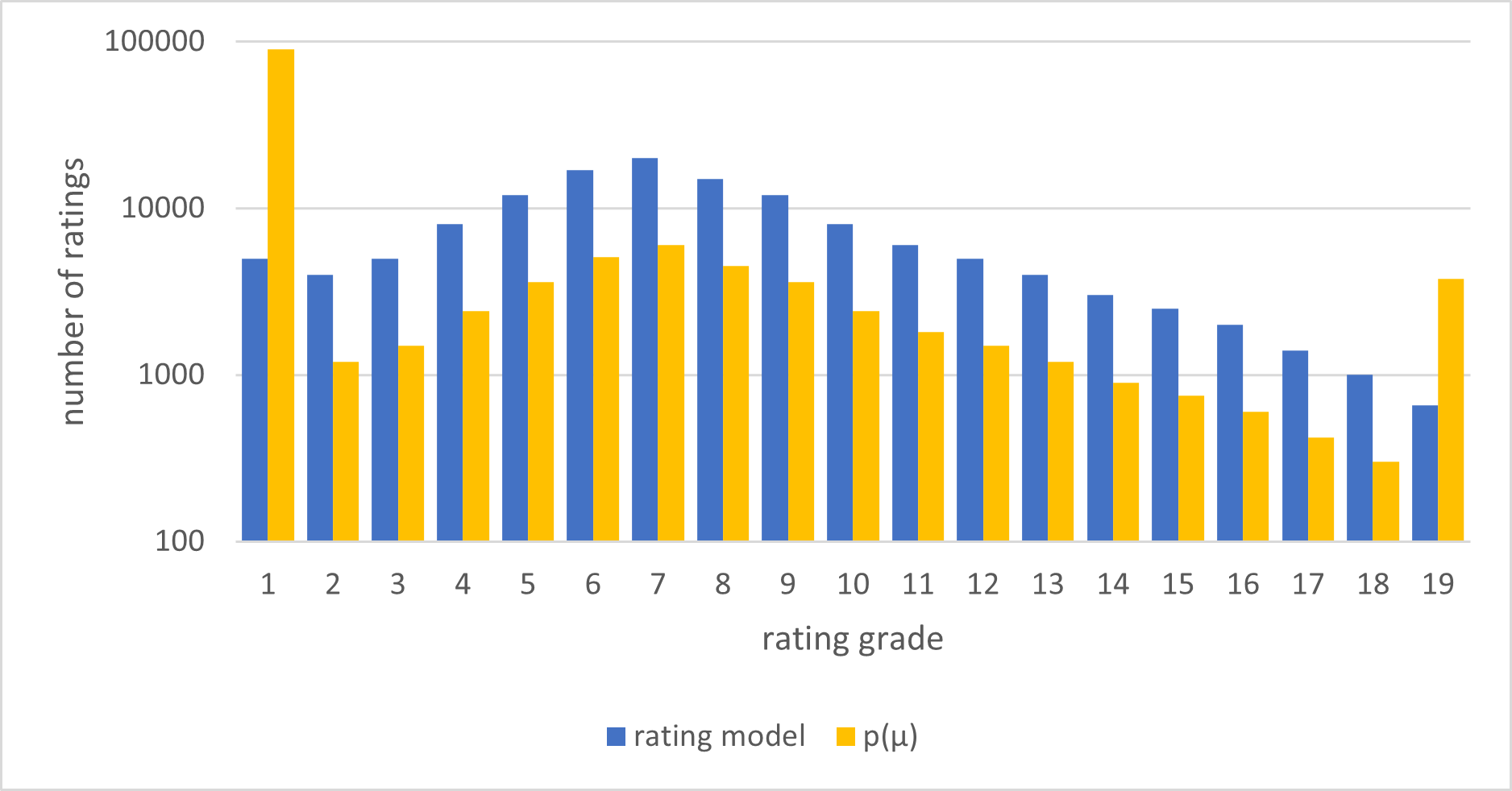}
\caption{Rating distribution of rating model and rating distribution after minimization with assumption of trusting the distribution of the model to a certain degree}
\label{Distribution 3}
\end{figure}

We conclude this section with a modified version of previously described approach using a uniform distribution among the $n$ worst rating grades. As before, we consider the $n$ worst rating grades $\left\{\PD_{m-n+1},...,\PD_m \right\}$ with $\PD_{m-n+1}>\mu$.\ Our aim is to formulate the requirement of a uniform distribution in such a way that we can again use the method for minimizing the variance already presented. To that end, we define
\[
\overline{\PD}:=\frac{1}{n}\sum_{i=m-n+1}^{m}\PD_i.
\]
We follow the idea from Section \ref{Statistical Test on Portfolio Level}, and proceed exactly as in the Appendix \ref{ Minimization problem} with $\overline \PD$ instead of $\PD_{\max}$. The effect on the variance or, almost equivalently, the acceptance range is illustrated in Figure \ref{Figure Modified uniform}.\ On the one hand, this approach avoids unnecessary high conservatism. On the other hand, it still sufficiently conservative and simple to implement. 

    \begin{figure}[htb]\label{Figure Modified uniform}
\includegraphics[width=\textwidth]{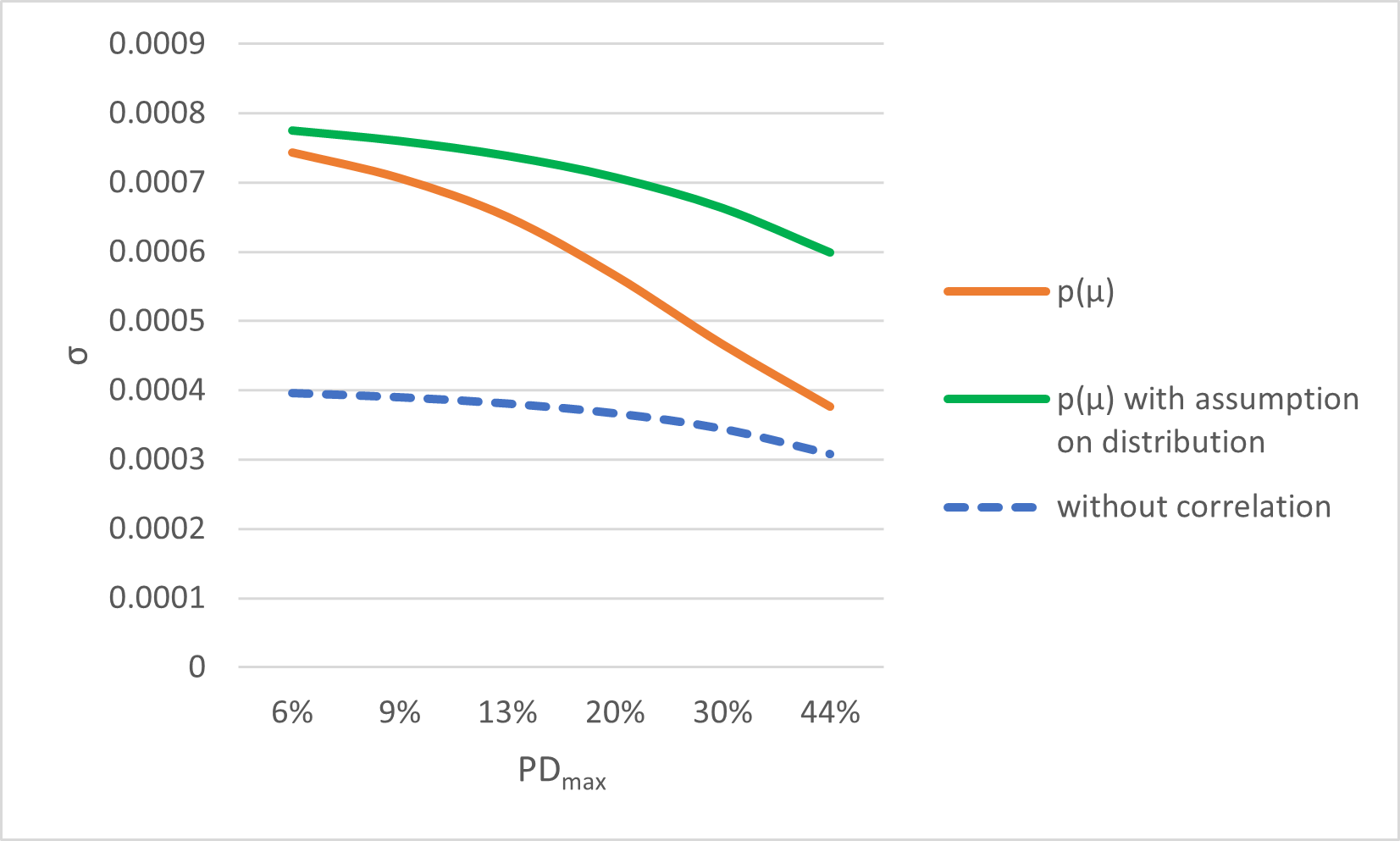}
\caption{Dependence of sigma on $\PD_{\max}$}
\label{PD_max Sigma}
\end{figure}

\subsection{Impact of Simplification on the Acceptance Range} \label{sc. reduction}
In this section we discuss, how different choices of affine linear functions for the simplification of the variance effect the acceptance range. We focus on the impact when choosing either the constant function $g=\PD_{\max}$ or identity function as in \eqref{eq.simplification}.

We start by computing the reduction of the acceptance range due to choosing the affine linear function as the constant $\PD_{\max}$ using a synthetic portfolio. We choose the number of customers per reference date to be constant with $n=1000$, let the number of reference dates be $N=60$, and $q=4$. We define $\PD_{\max}=0.2$, $\PD_{\min}=0.0003$, and set $\mu=0.01$. We recall that the components of the minimal solution $p(\mu)$ of \eqref{eq: reduced problem} consist only of  $\PD_{\max}$, $\PD_{\min}$ and $\mu_{k_0+1}$, where the latter is uniquely determined by the minimization problem. 

We now choose $g=\PD_{\max}$, i.e., the most conservative function possible. Proceeding exactly as in Section \ref{Statistical Test on Portfolio Level} and in the Appendix \ref{ Minimization problem}, we simply get different values for $\alpha_i$ and $c_i$ with $i=0,\ldots,q-1$, namely,
\[
\begin{array}{rcl}
    \alpha_i &= &\frac{q-i}{q}(1-\PD_{\max}) \\
   c_i  & =&0. \\
   
\end{array}
\]
Estimating the variance as in Section \ref{Statistical Test on Portfolio Level}, once with $g=\PD_{\max}$ and once with $g=\text{id}$, we get
\begin{equation}\label{eq quotient}
\frac{\left| K(\PD_{\max})-k(\PD_{\max}) \right|}{\left|K(\text{id})-k(\text{id})\right|}\approx 0.994,
\end{equation}
where $K(\PD_{\max})$ and $k(\PD_{\max})$ and $K(\text{id})$ and $k(\text{id})$ denote the upper and lower bound of the acceptance range for $g=\PD_{\max}$ and $g=\text{id}$, respectively. In this particular setting, we see that the choice of $g$ does not have a substantial influence on the size of the acceptance range.

The difference in the size of the acceptance range is largely influenced by the contributions of the obligors with $\PD=\PD_{\min}$.\ 
By Theorem \ref{Theorem}, large values of $\PD_{\max}$ lead to an increased number of pairs $(t,j)$ with $p_{t,j}=\PD_{\min}$.\ Hence, the size of the quotient in \eqref{eq quotient} is reduced for very large values of $\PD_{\max}$ and very small values of $\mu$. For instance, by redefining $\PD_{\max}=0.45$ and $\mu=0.0006$ we get 
\[
\frac{\left| K(\PD_{\max})-k(\PD_{\max}) \right|}{\left|K(\text{id})-k(\text{id})\right|}\approx 0.778.
\]
As a result, relevant differences in the acceptance ranges mostly occur in settings with very large values of $\PD_{\max}$ and values of $\mu$ that are close to $\PD_{\min}$.\ In practice, the maximally conservative choice of $g=\PD_{\max}$ is therefore not recommendable, since it assumes that on each reference date all persisting customers have been in the worst rating grade, independent of their current rating. Concluding, it is recommended to choose $g=\text{id}$ rather than $g=\PD_{\max}$.

\appendix

\section{Minimization Problem}\label{ Minimization problem}

Let $n\in \N$ and $\alpha_i \geq 0$ for $i=1,\ldots,n$. In this section, we compute the minimum of the function
\begin{equation} \label{Opt. Problem}
    f\colon \R^n \to \R,\quad x=\left(x_1,\ldots,x_n \right)\mapsto f(x):=\sum_{i=1}^{n}\alpha_i x_i
\end{equation}
 under side conditions
\begin{equation} \label{constraints App}
    \sum_{i=1}^{n}\beta_i x_i=m\quad \text{and}\quad c_1\leq x_i \leq c_2\quad \text{for }i= 1,\ldots,n
\end{equation}
with given constants $\beta_i >0$ for $i=1,\ldots,n$, $m>0$, and $0<c_1<c_2$.
We sort the ratios $\big(\alpha_i/\beta_i\big)_{i=1,\ldots,n}$ in descending order and denote by $\alpha_{(k)}$ and $\beta_{(k)}$ are the numerator and denominator of the $k$-th largest element of these ratios. In case of equal values among the ratios $\big(\alpha_i/\beta_i\big)_{i=1,\ldots,n}$, the one with smaller index $i$ is listed first.\
Moreover, $x_{(k)}$ denotes the variable of $f$ with coefficient $\alpha_{(k)}$.
We define
\[
k_0 :=\max \left\{k\: \left| \: \sum_{j=1}^{k}c_1\cdot \beta_{(j)}+\sum_{j=k+1}^{n}c_2\cdot \beta_{(j)}\geq m \right. \right\}.
\]
\begin{theorem} \label{Theorem}
The cost function in \eqref{Opt. Problem} is minimized under \eqref{constraints App} for 
\[
x_{(k)}=\begin{cases}
c_1,& k\leq k_0,\\
\mu_{k_0+1}, &k=k_0+1, \\
 c_2, &k_0+2\leq k\leq n,
\end{cases}
\]
where the value $\mu_{k_0+1}$ is uniquely determined by \eqref{constraints App} through 
\[
\mu_{k_0+1}=\frac{m-\sum_{j=1}^{k}c_1\cdot \beta_{(j)}-\sum_{j=k+2}^{n}c_2\cdot \beta_{(j)}}{\beta_{(k_0+1)}}.
\]
\end{theorem}

\begin{proof}
We show that for any $y\in \R^n$ that fulfills \eqref{constraints App}, $f(y)\geq f(x)$ holds if $y\neq x$.
Per construction of $x$, there exists some $\varepsilon>0$ such that either
\[
\sum_{j=1}^{k_0}y_{(j)}\cdot \beta_{(j)}=\sum_{j=1}^{k_0}x_{(j)}\cdot \beta_{(j)}+\varepsilon\quad\text{and}\quad
\sum_{j=k_0+1}^{n}y_{(j)}\cdot \beta_{(j)}=\sum_{j=k_0+1}^{n}x_{(j)}\cdot \beta_{(j)}-\varepsilon
\]
or
\[
\quad\sum_{j=1}^{k_0+1}y_{(j)}\cdot \beta_{(j)}=\sum_{j=1}^{k_0+1}x_{(j)}\cdot \beta_{(j)}+\varepsilon \quad\text{and}\quad \sum_{j=k_0+2}^{n}y_{(j)}\cdot \beta_{(j)}=\sum_{j=k_0+2}^{n}x_{(j)}\cdot \beta_{(j)}-\varepsilon
\]
holds. Without loss of generality we regard the first case. Let $\varepsilon=\sum_{i=1}^{k_0}\varepsilon_i$ such that 
\[
y_{(i)}=x_{(i)}+\frac{\varepsilon_i}{\beta_{(i)}}\quad \text{for all }i=1,...,k_0
\]
and $\varepsilon=\sum_{i=k_0+1}^{n}\varepsilon_i$ such that 
\[
y_{(i)}=x_{(i)}-\frac{\varepsilon_i}{\beta_{(i)}}\quad\text{for all }i=k_0+1,\ldots,n.
\]
Then,

\[
    \sum_{j=1}^{k_0}y_{(j)}\alpha_{(j)}=\sum_{j=1}^{k_0}\left(x_{(j)}+\frac{\varepsilon_j}{\beta_{(j)}}\right)\alpha_{(j)}
    \geq \sum_{j=1}^{k_0}x_{(j)}\alpha_{(j)}+ \frac{\alpha_{(k_0)}}{\beta_{(k_0)}} \varepsilon
\]
and 
\[
 \sum_{j=k_0+1}^{n}y_{(j)}\alpha_{(j)}=\sum_{j=k_0+1}^{n}\left(x_{(j)}-\frac{\varepsilon_j}{\beta_{(j)}}\right)\alpha_{(j)}
    \geq \sum_{j=k_0+1}^{n}x_{(j)}\alpha_{(j)}- \frac{\alpha_{(k_0+1)}}{\beta_{(k_0+1)}} \varepsilon.
\]
Hence,
\[
\sum_{j=1}^{n}y_{(j)}\alpha_{(j)}\geq \sum_{j=1}^{n}x_{(j)}\alpha_{(j)}+\left(\frac{\alpha_{(k_0)}}{\beta_{(k_0)}} -\frac{\alpha_{(k_0+1)}}{\beta_{(k_0+1)}}\right) \varepsilon\geq \sum_{j=1}^{n}x_{(j)}\alpha_{(j)}
\]
For the second case, one proceeds exactly in the same way.
\end{proof}

We now translate the previous theorem into the setting of Section \ref{Statistical Test on Portfolio Level}. There, we aim to minimize the sum
\[
\sum_{t=1}^{N}\sum_{j\in \Lambda_{t}}^{}\alpha_{t,j}p_{t,j}
\]
under the side condition
\begin{equation*}
    \mu=\frac{1}{N}\cdot \sum_{t=1}^{N}\frac{1}{n_t}\sum_{j\in \Lambda_{t}}^{}p_{t,j}\text{ and }p_{t,j}\in[\PD_{\min} ,\PD_{\max} ].
\end{equation*}
We rewrite 
\[
\mu=\sum_{t=1}^{N}\sum_{j\in \Lambda_{t}}^{} \beta_{t,j} p_{t,j}
\]
with $\beta_{t,j}:=\frac{1}{N}\cdot \frac{1}{n_t}$.
We first sort the set of all possible tuples $(t,j)$ in ascending order, first with respect to $j$ and then with respect to $t$, i.e., we can identify any possible tuple $(t,j)$ with a natural number $i=1,\ldots,n:=\sum_{t=1}^{N}n_t$, which simplifies the cost function to
\begin{equation} 
    f\colon \R^n \to \R,\quad p=\left(p_1,...,p_{n} \right)\mapsto f(p):=\sum_{i=1}^{n}\alpha_i p_i
\end{equation}
with side conditions
\begin{equation} 
    \sum_{i=1}^n\beta_i p_i=\mu\quad \text{and}\quad \PD_{\min}\leq p_i \leq \PD_{\max}\quad \text{for  }i= 1,\ldots,n.
\end{equation}
Now, we are exactly in the setting of \eqref{Opt. Problem} and \eqref{constraints App}, and can apply Theorem \ref{Theorem}.

\section{Test on Portfolio Level without Solving Minimization Problem}\label{Appendix E}

We start by estimating the covariance between two default rates on reference dates $\RD_{t}$ and $\RD_{s}$.\ Using Equation \eqref{eq:3}, we find that
\begin{align}  \label{eq:Cov(X,Y)2} 
\Cov\left(X_{t},X_{s} \right)
&=\frac{1}{n_{t}\cdot n_{s}}\sum_{j\in \Lambda_{t} \cap \Lambda_{t_2}}^{}\Cov\left(x_{t,j},x_{s,j} \right) \notag \\
&\geq \frac{1}{n_{t}\cdot n_{s}}\sum_{j\in \Lambda_{t} \cap \Lambda_{s}}^{}p_{s,j} w_{t,s}\left(1-\PD_{\max} \right) \notag \\
&=\frac{k_{t,s}}{n_{t}\cdot n_{s}}w_{t,s}\left(1-\PD_{\max} \right) \frac{1}{k_{t,s}} \sum_{j\in \Lambda_{t} \cap \Lambda_{s}}^{}p_{s,j} \notag \\
&\geq \frac{k_{t,s}}{n_{t}\cdot n_{s}}w_{t,s}\left(1-\PD_{\max} \right) \gamma \cdot \E (X_{s}).    
\end{align}
The parameter $\gamma >0$ is less or equal to the ratio of the average default risk of persisting customers in the portfolio and the average default risk of the entire portfolio, i.e.,
$$\E\bigg(\frac{1}{k_{t,s}} \sum _{j\in \Lambda_{t}\cap \Lambda_{s}}^{}x_{s,j}\bigg)\geq\gamma\cdot \E\bigg(\frac{1}{n_{s}}\sum _{j\in \Lambda_{s}}^{}x_{s,j}\bigg).$$ 
In view of the fact that persisting customers should be the majority in the portfolio, it is reasonable to assume that $\gamma\approx1$.\ The value is to be estimated conservatively depending on the portfolio under consideration, i.e., in case of doubt, it should be chosen relatively small. 
If a portfolio consists only of persisting customers, $\gamma$ is always $1$, while the parameter for portfolios with a small proportion of persisting customers can be above or below $1$.\ In the following, the parameter $\gamma$ is to be understood as an estimate that is independent of the reference date and conservative for the entire history. From Equation \eqref{Optimization}, we see that
\begin{align}\label{eq: estimate with C}
   \sigma^2\geq \frac{1}{N^2}\sum_{t=1}^N \frac{1}{n_t^2}\sum_{j\in \Lambda_t} p_{t,j}(1-p_{t,j})\geq C\cdot\mathbb{E}\bigg(\frac{1}{N}\sum_{t=1}^{N}X_t \bigg) 
\end{align}
with $C=\frac{1-\PD_{\max}}{n_{\max}}$.\ Hence, using Equation \eqref{eq:Cov(X,Y)},
\begin{align*}
 N^2\cdot\sigma^2&=\sum _{t=1}^{N}{\sigma _i^2}+2\sum_{i=1}^{q-1}\sum _{t=1}^{N-i}\Cov\left (X_t,X_{t+i}\right) \\
 &\geq \sum_{t=1}^{N}\sigma_t^2+2\sum_{i=1}^{q-1}\sum_{t=1}^{N-i}\frac{k_{t,t+i}}{n_t\cdot n_{t+i}}\cdot\frac{q-i}{q}\left(1-\PD_{\max} \right)\cdot\gamma\cdot\mu_{t+i} \\
 &\geq \sum_{t=1}^{N}\sigma_t^2+2\sum_{i=1}^{q-1}  \min_{t=1,\ldots,N-i}{\left(\frac{k_{t,t+i}}{n_t\cdot n_{t+i}}\right)}\frac{q-i}{q}\left(1-\PD_{\max} \right)\cdot\gamma\sum_{t=1}^{N-i}\mu_{t+i} \\
 &\geq \sum_{t=1}^{N}\sigma_t^2+K_1\sum_{t=1}^{N}\mu_t-K_2
\end{align*}
with
\begin{align*}
    K_1&:=2\sum_{i=1}^{q-1}  \min_{t=1,\ldots,N-i}{\left(\frac{k_{t,t+i}}{n_t\cdot n_{t+i}}\right)}\frac{q-i}{q}\left(1-\PD_{\max} \right)\cdot\gamma\quad \text{and}\\
    K_2&:=2\sum_{i=1}^{q-1}  \min_{t=1,\ldots,N-i}{\left(\frac{k_{t,t+i}}{n_t\cdot n_{t+i}}\right)}\frac{q-i}{q}\left(1-\PD_{\max} \right)\cdot\gamma \cdot i \cdot \mu_{\old},
\end{align*}
where $\mu_{\old}\geq\max_{i=1,\ldots,q-1}{\mu_i}$.\ The term $\mu_{\old}$ cannot be determined by the null hypothesis, thus a conservative value is to be determined.\ For portfolios consisting of customers with good credit ratings, the choice $\mu_{\old}=\PD_{\max}$ is clearly not adequate. Using Equation \eqref{eq: estimate with C}, it follows that
$$N^2\cdot\sigma^2\geq c\sum_{i=1}^{N}\mu_i+K_1\sum_{i=1}^{N}\mu_i-K_2=N\left(C +K_1\right)\mu-K_2$$
with $c\geq\frac{1}{n_{\max}}\left(1-\PD_{\max}\right)$ and
\begin{align} \label{est: sigma}
\sigma^2 \geq \sigma_{{\rm alt}}^2(\mu):= \frac{1}{N}\left(c+K_1 \right)\mu-\frac{K_2}{N^2}.
\end{align}
Moreover, for portfolios with a short history and low PDs, it is not always possible to choose $\mu_{\old}=\PD_{\max}$ since $\frac{1}{N}\left(c+K_1 \right)\mu-\frac{K_2}{N^2}$ could become smaller than $0$. In order to be able to rule out this case, which also appears to be of little interest for applications of the test, we set an upper bound for $\mu_{\old}$, i.e.,
\[
\mu_{\old}\leq \frac{N\cdot \left( c+K_1\right)\mu}{(q-1)\cdot K_1}.
\]
Considering the random variable $Z^*\sim \mathcal{N} \left(\LTCT,\sigma_{{\rm alt}}^2(\LTCT)\right) $, we define the limits of the acceptance range as
$$k:=\LTCT+\Phi^{-1} \left(\frac{1}{2}\alpha \right) \sigma_{{\rm alt}}\big(\LTCT \big)$$
and
$$K:=\LTCT+\Phi^{-1} \left(1-\frac{1}{2}\alpha \right) \sigma_{{\rm alt}}\big(\LTCT \big).$$
The calibration test passes if
$$\LTDR\in\left[k,K\right].$$



\end{document}